\documentclass[%
 reprint,
superscriptaddress,
 amsmath,amssymb,
 aps,
prb,
longbibliography,
]{revtex4-2}

\usepackage{graphicx}
\usepackage{dcolumn}
\usepackage{bm}
\usepackage{comment}
\usepackage{multirow}
\usepackage{color}

\usepackage{hyperref}

\usepackage{amsthm}
\theoremstyle{plain}

\newtheorem*{thm*}{Theorem~\ref{t:main}}
\newtheorem{lem}{Lemma}
\newtheorem*{lem*}{Lemma}
\newtheorem{theorem}[lem]{Theorem}
\theoremstyle{definition}
\newtheorem{dfn}{Definition}
\newtheorem*{dfn*}{Definition}

\newcommand{\bmat}[1]{\begin{matrix} #1 \end{matrix}}
\newcommand{\loq}[1]{\begin{array}{l} (#1) \end{array}}

\newcommand{\eq}[1]{\begin{eqnarray} #1 \end{eqnarray}}
\newcommand{\nx}{ \nonumber \\}

\def \({\left(}
\def \){\right)}
\def \[{\left[}
\def \]{\right]}

\newcommand{\sumi}{\sum_{i=1}^N}

\newcommand{\jx}{J_X X_i X_{i+1}}
\newcommand{\jy}{J_Y Y_i Y_{i+1}}
\newcommand{\jz}{J_Z Z_i Z_{i+1}}
\newcommand{\hx}{h_X X_i}
\newcommand{\hy}{h_Y Y_i}
\newcommand{\hz}{h_Z Z_i}
\newcommand{\stard}{* \cdots *}
\newcommand{\stara}{& * & \cdots & * &}
\newcommand{\starn}{&   &        &   &}

\newcommand{\boldA}{\boldsymbol{A}}
\newcommand{\boldB}{\boldsymbol{B}}

\newcommand{\ak}{{\boldA}^k}
\newcommand{\aki}{\ak_i}
\newcommand{\qa}{q_{\aki}}
\newcommand{\Bk}{{\boldB}^{k+1}}
\newcommand{\Bki}{\Bk_i}
\newcommand{\rB}{r_{\Bki}}
\newcommand{\bk}{{\boldB}^k}
\newcommand{\bki}{\bk_i}
\newcommand{\rb}{r_{\bki}}

\begin{document}

\preprint{APS/123-QED}

\title{Proof of the absence of local conserved quantities\\ in general spin-1/2 chains with symmetric nearest-neighbor interaction}
\author{Mizuki Yamaguchi}
\email{yamaguchi-q@g.ecc.u-tokyo.ac.jp}
\affiliation{
Graduate School of Arts and Sciences, The University of Tokyo,
3-8-1 Komaba, Meguro, Tokyo 153-8902
, Japan}
\author{Yuuya Chiba}
\email{yuya.chiba@riken.jp}
\affiliation{Nonequilibrium Quantum Statistical Mechanics RIKEN Hakubi Research Team,
RIKEN Cluster for Pioneering Research (CPR), 2-1 Hirosawa, Wako, Saitama 351-0198, Japan}
\author{Naoto Shiraishi}
\email{shiraishi@phys.c.u-tokyo.ac.jp}
\affiliation{
Graduate School of Arts and Sciences, The University of Tokyo,
3-8-1 Komaba, Meguro, Tokyo 153-8902
, Japan}

\begin{abstract}

We provide a rigorous proof of the absence of nontrivial local conserved quantities in all spin-1/2 chains with symmetric nearest-neighbor interaction, except for known integrable systems.
This result shows that there are no further integrable system that awaits to be discovered.
Our finding also implies that there is no intermediate systems with a finite number of nontrivial local conserved quantities.
In addition, we clarify all short-support conserved quantities in non-integrable systems, which we need to take into account in analyses of thermalization and level statistics.

\end{abstract}
\maketitle


\section{Introduction}\label{Introduction}
Quantum integrability is an important subject in the field of mathematical physics \cite{baxter2016exactly, jimbo1994algebraic, takahashi1999thermodynamics, fagotti2013reduced}.
Various elaborated tools for solving integrable systems have been developed \cite{sklyanin1992quantum, faddeev1996algebraic, tetel1982lorentz, thacker1986corner}, and many models are proven to be integrable \cite{1931.Bethe.ZP.71, lieb1961two, yang1966one, lieb1968absence, baxter1971one, maassarani1998xxc, kitaev2001unpaired, yanagihara2020exact, sutherland1975model, Takhtajan:1982jeo, babujian1982exact, barber1989spectrum}.
One prominent useful tool is the quantum inverse scattering method \cite{sklyanin1992quantum}, which is extremely close to the algebraic Bethe ansatz \cite{faddeev1996algebraic}.
The quantum inverse scattering method solves a wide class of integrable systems by constructing a transfer matrix $T(\lambda)$, from which a sequence of local conserved quantities $\{Q_k\}$ are generated as the high-order derivatives of $\log T(\lambda)$. 
At the same time, the same transfer matrix yields Bethe ansatz eigenstates of integrable systems, which is an outline of the algebraic Bethe ansatz.
As seen from the above construction, an infinite sequence of local conserved quantities is the central ingredient of solvability.

Unlike the deep understanding of integrable systems, very few things have been rigorously established on non-integrable systems.
Here, we identify non-integrable systems as systems with no nontrivial local conserved quantity \cite{mori2018thermalization, shiraishi2019proof}.
Of course, there are many models which cannot be solved by any known tools, and therefore these models are strongly expected to be non-integrable.
This expectation is supported by energy level spacing statistics, where the distribution tends to be the Poisson distribution in integrable systems and the Wigner-Dyson distribution in non-integrable systems \cite{dyson1962statistical, wigner1993characteristic, rigol2009breakdown, santos2010onset, atas2013distribution}.
However, this expectation has been stuck at a mere expectation without mathematical justifications for a long time.
One crucial reason is that present tools for integrable systems presume specific forms of local conserved quantities, and there has been no good idea to generally exclude all possible forms of local conserved quantities.
From this background, some mathematical physicists even felt that non-integrability is out of the scope of mathematical physics, which can be just accepted as a plausible description.

Contrary to such pessimistic views, recent analytical studies have revealed that non-integrability is indeed in the scope of mathematical physics.
The XYZ model with a magnetic field \cite{shiraishi2019proof}, the transverse-field Ising model \cite{chiba2024proof}, the PXP model \cite{park2024proof}, and the Heisenberg model with next-nearest-neighbor interaction \cite{shiraishi2024absence} were rigorously proven to have no nontivial local conserved quantities. 
These studies pinpoint models expected to be non-integrable and overthrow one by one.
Thus, despite notable achievements of these studies, whether overlooked integrable systems still exist in a class of some simple systems (e.g., systems with nearest-neighbor interaction) is left as an open problem.

This problem, which is theoretically important by itself, is also related to our observation of generic macroscopic systems in our world.
We empirically know that generic macroscopic systems follow thermalization to the equilibrium state \cite{neumann1929beweis, goldstein2010long}, the Kubo formula in linear response theory \cite{kubo1957statistical}, and the Fourier law in heat conduction \cite{fourier1822}. 
However, integrable systems with many local conserved quantities are known not to follow these empirical laws \cite{mazur1969non, suzuki1971ergodicity, zotos1997transport, saito2003strong, rigol2008thermalization}.
Thus, integrable systems are considered to be exceptional, and most of systems should be non-integrable.
This view is highly plausible, though a mathematical foundation again lacks.

In this study, we conduct a comprehensive analysis of local conserved quantities in general spin-1/2 chains with shift-invariant and parity-symmetric nearest-neighbor interactions. We rigorously prove that all systems in this class, aside from known integrable systems, do not possess nontrivial local conserved quantities. This class includes integrable systems of two different origins: models solved by the Bethe ansatz (the XXX model \cite{1931.Bethe.ZP.71}, the XXZ model \cite{yang1966one} and the XYZ model \cite{baxter1971one}) and models mapped to free fermion systems (the XY model \cite{lieb1961two} and the transverse-field Ising model \cite{lieb1961two}). 
Our result establishes that there is no overlooked integrable systems in this class, and known integrable systems serve as the complete list of all integrable systems (Table~\ref{table:integrable}). 
This fact suggests the expected ubiquitousness of non-integrability.
Furthermore, as an important implication of our result, any model in this class is shown to have either (i) infinitely many local conserved quantities or (ii) no nontrivial local conserved quantity, and no system has a finite number of nontrivial local conserved quantities.

The structure of this paper is as follows. In Sec.~\ref{sec:Setup}, we introduce the setup and main claim, with describing the class of systems studied in this paper. We also define local conserved quantities to clarify the statement of the main theorem. In Sec.~\ref{sec:Integrable}, we list known integrable systems. 
In Sec.~\ref{sec:Validity}, we review previous discussions on the definitions of integrability and non-integrability and declare our tentative characterization.
Sec.~\ref{sec:Proof} is the main part of this paper, where we present the proof of the main theorem. We divide cases by the rank of interaction coefficients and treat these cases separately.
In Sec.~\ref{sec:Trivial} we explore ``trivial'' local conserved quantities in non-integrable systems. Finally, in Sec.~\ref{sec:Conclusion}, we summarize our research findings and discuss their implications.

\section{Setup and main result}\label{sec:Setup}

We consider a general spin-1/2 chain on $N$ sites, with nearest-neighbour interactions and magnetic fields, which is shift invariant and parity symmetric.
The general form of the Hamiltonian is expressed as
\eq{
H =  && \sumi \begin{pmatrix} X_i & Y_i & Z_i \end{pmatrix}
\begin{pmatrix}
J_{XX} & J_{XY} & J_{XZ} \\
J_{YX} & J_{YY} & J_{YZ} \\
J_{ZX} & J_{ZY} & J_{ZZ} \\ \end{pmatrix}
\begin{pmatrix} X_{i+1} \\ Y_{i+1} \\ Z_{i+1} \end{pmatrix} \nx
&+& \sumi \begin{pmatrix} h_X & h_Y & h_Z \end{pmatrix} \begin{pmatrix} X_i \\ Y_i \\ Z_i \end{pmatrix}
\label{general} ~,
}
where $X_i,Y_i,Z_i$ represent the Pauli matrices $\sigma^x_i, \sigma^y_i, \sigma^z_i$, respectively,
and $J_{\alpha \beta}$ and $ h_\alpha$ ($\alpha, \beta \in \{X,Y,Z\}$) are arbitrary symmetric real numbers, i.e., $J_{\alpha\beta}=J_{\beta\alpha}$ for any $\alpha, \beta \in \{X,Y,Z\}$.
We identify site $N+1$ to site $1$, meaning the periodic boundary condition.

In order to state our claim in a rigorous manner, we first define some terms.
\begin{dfn}
    An operator is a {\it $k$-support operator} if it acts on $k$ consecutive sites.
\end{dfn}
\begin{dfn}
An operator is a {\it $k$-support quantity} if it can be written by a sum of $l$-support operators with $l \leq k$ and cannot be expressed by a sum of those with $l \leq k-1$.
\end{dfn}
For instance, $X_i X_{i+1} Y_{i+2}$ is a $3$-support operator and $Z_i X_{i+4}$ is a $5$-support operator.
Both $\sumi Z_i$ and $\sumi (-1)^i Z_i$ are $1$-support quantities.
A $k$-support quantity is not necessarily a shift sum of $k$-support operators.

\begin{dfn}
    Given a system and its Hamiltonian $H$, an operator is a {\it $k$-support conserved quantity} if it is commutative with $H$ and is a $k$-support quantity.
\end{dfn}

We will refer to a $k$-support conserved quantities of $k=O(1)$ as local conserved quantities.
We characterize non-integrability by the absence of nontrivial local conserved quantities.
Here, we call $k$-support conserved quantities with $k \geq 3$ non-trivial local conserved quantities, to distinguish them from trivial local conserved quantities such as its own Hamiltonian.
Prior studies \cite{mori2018thermalization, shiraishi2019proof, chiba2024proof, park2024proof, shiraishi2024absence} have characterized non-integrability as the absence of nontrivial local conserved quantities, and this study follows the same definition.
The validity of this characterization is discussed in Sec.~\ref{sec:Validity}.

Now we are ready to state our Theorem~\ref{t:main}, which is the main theorem of this paper.
We claim that there is no unknown integrable system expressed as Eq.~\eqref{general}, and all the remaining systems not shown to be integrable are indeed non-integrable in the above sense.

\begin{theorem}\label{t:main}
Any system represented by Eq.~\eqref{general} has no $k$-support conserved quantity with $3 \leq k \leq N/2$, except for known integrable systems listed in Table~\ref{table:integrable} and their equivalents.
\end{theorem}
This theorem establishes an important fact that there is no further integrable systems beyond known ones in spin-$1/2$ chains, as long as the system is shift invariant, parity symmetric, and nearest-neighbor interacting.
This result confirms that integrable systems are exceptional and almost all systems are non-integrable, at least in the class described as Eq.~\eqref{general}.
Note that the upper bound of $k$ is almost tight, because $H^2$ is a $(\lfloor N/2 \rfloor + 2)$-support conserved quantity. 
We also note that 2-support and 1-support conserved quantities, which are ignored here as trivial, are discussed in Sec.~\ref{sec:Trivial}.

Theorem~\ref{t:main} also highlights a sharp contrast between integrable systems and other systems.
All known integrable systems have a $k$-support conserved quantity for each $k$,
which are mutually commutative. 
In contrast, all other systems described in Eq.~\eqref{general} have no $k$-support conserved quantities for all $3 \leq k \leq N/2$.
That is, systems represented by Eq.~\eqref{general} have only two extreme types of systems, one has infinite local conserved quantities and the other has no nontrivial local conserved quantities, and there is no intermediate system with a finite number of nontrivial local conserved quantities.

\section{Known integrable systems}\label{sec:Integrable}
Before going to the proof of Theorem~\ref{t:main}, we briefly review known integrable systems and the notion of quantum (non-)integrability in this and next section.
In this section, we present all known integrable systems described by Eq.~\eqref{general} in order to clarify the precise statement of Theorem~\ref{t:main}.
Here we list in Table~\ref{table:integrable} all known integrable systems.
Theorem~\ref{t:main} tells that Table~\ref{table:integrable} with a global spin transformation in fact comprises all integrable systems represented by Eq.~\eqref{general}.

\begin{table*}
\caption{List of known integrable systems in Eq.~\eqref{general}}
\label{table:integrable}
\def\arraystretch{1.5}
\begin{tabular}{|c|c|c|c|}
\hline
Method & Name & Hamiltonian $(J_X,J_Y,J_Z \neq 0)$ & Reference \\ \hline
\multirow{3}{*}{Bethe ansatz} & XXX (Heisenberg) & $\sumi 
( J_X X_i X_{i+1} + J_X Y_i Y_{i+1} + J_X Z_i Z_{i+1}
+ h Z_i)$ & \cite{1931.Bethe.ZP.71} \\ \cline{2-4}
 & XXZ & $\sumi 
( J_X X_i X_{i+1} + J_X Y_i Y_{i+1} + J_Z Z_i Z_{i+1}
+ h Z_i)$ & \cite{yang1966one} \\ \cline{2-4}
 & XYZ & $\sumi
( J_X X_i X_{i+1} + J_Y Y_i Y_{i+1} + J_Z Z_i Z_{i+1})$ & \cite{baxter1971one} \\ \hline
\multirow{2}{*}{Free fermions} & XY & $\sumi 
( J_X X_i X_{i+1} + J_Y Y_i Y_{i+1}
+ h_Z Z_i)$ & \cite{lieb1961two} \\ \cline{2-4}
 & Transverse-field Ising & $\sumi 
( J_Z Z_i Z_{i+1}
+ h_X X_i)$ & \cite{lieb1961two} \\ \hline
\multirow{2}{*}{(trivial)}  & Longitudinal-field Ising & $\sumi 
( J_Z Z_i Z_{i+1} + h_Z Z_i)$ & - \\ \cline{2-4}
 & Free spins & $\sumi 
h_Z Z_i$ & - \\ \hline
\end{tabular}
\end{table*}

Bethe-ansatz type integrable systems, including the XXX model, the XXZ model, and the XYZ model,
are solved by the quantum inverse scattering method \cite{sklyanin1992quantum} and algebraic Bethe ansatz \cite{faddeev1996algebraic}.
The XXX model is also called (isotropic) Heisenberg model,
and the XXZ model and the XYZ model are referred to as anisotropic Heisenberg models.

Free-fermion type integrable systems, including the XY model and the transverse-field Ising model,
can be mapped to free fermion systems by the Jordan-Wigner transformation \cite{Jordan1928berDP}.

The longitudinal field Ising model and the free spin model are classical spin systems. Thus, it is obvious that they have infinite local conserved quantities.

These integrable systems, in contrast to non-integrable systems, possess at least one $k$-support conserved quantity for any $k$. As an example, we here describe $k$-support conserved quantities of the Heisenberg (XXX) model for $k=3,4,5$ by following Ref.~\cite{grabowski1994quantum}. 
The $3$-support conserved quantity can be written as 
\eq{
Q_3 = \sumi \loq{
    X_i Y_{i+1} Z_{i+2} + Y_i Z_{i+1} X_{i+2} + Z_i X_{i+1} Y_{i+2} \\
   -Y_i X_{i+1} Z_{i+2} - Z_i Y_{i+1} X_{i+2} - X_i Z_{i+1} Y_{i+2}
}~. \nx
}
To simplify expression, we adopt a vector notation for the Pauli operator as ${\bm \sigma}_i = \begin{pmatrix}
X_i \\ Y_i \\ Z_i
\end{pmatrix}$.
Using this symbol, the $3$-support conserved quantity is rewritten as
\eq{
Q_3 
&=& \sumi({\bm \sigma}_i \times {\bm \sigma}_{i+1}) \cdot {\bm \sigma}_{i+2} ~,
}
where $\times$ represents the cross product.
Similarly, the $4$-support and $5$-support conserved quantities are expressed as 
\eq{
Q_4 &=& \sumi (
2 ( ({\bm \sigma}_i \times {\bm \sigma}_{i+1}) \times {\bm \sigma}_{i+2}) \cdot {\bm \sigma}_{i+3}
+ {\bm \sigma}_i \cdot {\bm \sigma}_{i+2}  
) \nx && -4H~, \\
Q_5 &=& \sumi \loq{ 6 ((({\bm \sigma}_i \times {\bm \sigma}_{i+1}) \times {\bm \sigma}_{i+2}) \times {\bm \sigma}_{i+3}) \cdot {\bm \sigma}_{i+4} \\
+ ({\bm \sigma}_i \times {\bm \sigma_{i+2}}) \cdot {\bm \sigma}_{i+3}
+ ({\bm \sigma}_i \times {\bm \sigma_{i+1}}) \cdot {\bm \sigma}_{i+3}
} \nx && -18Q_3 ~.
}
In Bethe-ansatz type integrable systems, the quantum inverse scattering method \cite{sklyanin1992quantum} provides an infinite sequence of $k$-support conserved quantities. For the XXZ model, $k$-support conserved quantities with general $k$ are provided in Refs.~\cite{grabowski1994quantum, nozawa2020explicit}, and those of the XYZ model are calculated in Ref.~\cite{grabowski1995structure, fukai2023all}. Free-fermion type integrable systems also have $k$-support conserved quantities for general $k$ \cite{grabowski1995structure}.

\section{Comparison of characterizations of non-integrability}
\label{sec:Validity}
As already mentioned, there is no unique established definition of quantum integrability and non-integrability, and various related but inequivalent characterizations have been employed in literature \cite{mori2018thermalization, caux2011remarks, gogolin2016equilibration, shiraishi2019proof}.
In this section, we briefly discuss why we cannot define quantum integrability and non-integrability easily in contrast to the classical case. 
At the same time, we also compare several characterizations of quantum integrability and non-integrability.

In classical Hamiltonian systems, there exists an established definition of integrability known as the Liouville integrability \cite{Liouville1855}, stating that a system is integrable if there are sufficiently many conserved quantities which are commutative with each other in the sense of the Poisson bracket.
The Liouville-Arnold theorem \cite{arnol1963small} states that the above definition of integrability is equivalent to the property that the system can be solved in quadratures.
Hence, in classical systems, we have two well-defined characterizations of integrability, many conserved quantities and solvability, both of which are in fact equivalent.

In contrast to the classical case, the quantum integrability and non-integrability are uneasy to define.
The difficulty comes from the fact that all quantum systems are formally solvable (by diagonalizing the Hamiltonian), and for a similar reason all quantum systems have infinitely many conserved quantities.
More precisely speaking on the latter, for the set of eigenstates $\{ |\psi_\alpha \rangle \}$ ($\alpha = 1,2,\dots 2^N$ for spin-$1/2$ chains),
all $2^N$ projection operators $|\psi_\alpha \rangle \langle \psi_\alpha |$ are conserved quantities commutative with each other.
Thus, to obtain a meaningful characterization of integrability, we need to restrict the class of solutions (when focusing on solvability) and the class of conserved quantities (when focusing on the number of2 conserved quantities).

In the case of short-range interaction Hamiltonians, one widely used characterization of quantum integrability and non-integrability is to restrict conserved quantities to local ones \cite{mori2018thermalization, shiraishi2019proof}, which we have already adopted in Sec.~\ref{sec:Setup} and Sec.~\ref{sec:Integrable}.
In this characterization, a system is said to be integrable if the system has infinitely many local conserved quantities in the thermodynamic limit.
In addition, in this paper we employ a strict characterization of non-integrability stating that a system is non-integrable if the system has no nontrivial local conserved quantity.
We emphasize that in this characterization non-integrability is not the negation of integrability, and some system may be neither integrable and non-integrable.
This happens if the system has a finite number of nontrivial local conserved quantities.
However, as we will prove, such intermediate system does not exist in the class of Hamiltonians written in Eq.~\eqref{general}.

The presence of infinitely many local conserved quantities has a strong connection to solvability.
In integrable systems, the eigenstates are specified by quantum numbers of the obtained infinite number of local conserved quantities. 
In fact, the quantum inverse scattering method, which is a standard method to solve quantum integrable systems related to the algebraic Bethe ansatz, provides an infinite number of local conserved quantities in the course of solving this system \cite{sklyanin1992quantum, faddeev1996algebraic}.

We remark that when we say infinitely many local conserved quantities, we implicitly assume that the number of local conserved quantities is sufficient for specifying the eigenstates.
Strictly speaking, there is a gap between being infinite in the thermodynamic limit and being sufficient to solve.
There is a known system whose number of local conserved quantity is infinite but insufficient \cite{hamazaki2016generalized}.

\section{Proof of Theorem~\ref{t:main}}\label{sec:Proof}

\subsection{Preliminary treatment}
\label{subsec:pt}
As preliminary to the proof of Theorem~\ref{t:main}, we first apply a {\it global spin rotation} to modify the form of the Hamiltonian \eqref{general} with keeping the generality.
The global spin rotation changes the axis of the Pauli operators on all sites simultaneously.
Namely, we apply a site-independent rotation matrix $R \in SO(3)$ to three Pauli matrices as
$\begin{pmatrix}
X_i' \\ Y_i' \\ Z_i'
\end{pmatrix}
= R 
\begin{pmatrix}
X_i \\ Y_i \\ Z_i
\end{pmatrix} ~,
$ and rename $X',Y',$ and $Z'$ as the Pauli matrices $X,Y,$ and $Z$, respectively.
Remark that the energy spectrum and the locality of conserved quantities are invariant to this transformation.

Since the system is parity symmetric, the interaction matrix $J = (J_{\alpha \beta} )$ in Eq~\eqref{general} is a (real) symmetric matrix, which can be diagonalized by a global spin rotation. 
After diagonalization, any Hamiltonian in the form of Eq.~\eqref{general} is reduced to the following standard form
\eq{
H = \sumi \loq{  \jx  +\jy +\jz \\ +\hx +\hy +\hz } ~,
\label{standard}
}
where we rename $J_{XX},J_{YY},$ and $J_{ZZ}$ as $J_X, J_Y,$ and $J_Z$, respectively.
Thus, it suffices to classify integrability and non-integrability of the system in the standard form \eqref{standard} only.

We call the number of non-zero elements in $J_X$, $J_Y$, and $J_Z$ as the rank of $J$, which plays an important role in the proof. 
Precisely, we will prove Theorem~\ref{t:main} with dividing case by the rank of $J$.

\subsection{Proof idea: Expansion of quantities.}
A key idea of the proof of the absence of local conserved quantities is expansion of a general local quantity with a Pauli basis.
Since any $2\times 2$ Hermitian matrix can be expanded by three Pauli matrices $X$, $Y$, and $Z$, and the identity $I$, a sequence of Pauli matrices serves as a basis of any observable on $S=1/2$ chains.
We introduce a {\it $k$-support Pauli string} which is a string of Pauli matrices and identity whose nontrivial consecutive support is $k$ sites \footnote{
Note that the specification of support is unique for $k(\leq N/2)$-support Pauli strings, under the rule that support is chosen as the shortest one.}.
We denote by $G^k$ the set of all $k$-support Pauli strings with no $I$ at the both ends.
Examples are $G^1 = \{X,Y,Z\}$, $G^2 = \{ XX, XY, \dots, ZY, ZZ\}$(9 elements), and $G^3 = \{ XXX, XXY, \dots , ZIY, ZIZ \}$(36 elements).

By expanding a candidate conserved quantity $Q$, the Hamiltonian $H$, and their commutator $[Q,H]$ with this Pauli strings, the conservation condition $[Q,H]=0$ gives a system of linear equations of the expansion coefficients of $Q$. 
We prove non-integrability by showing that no coefficient satisfies these equations except trivial solutions.

Precisely, any $k$-support quantity $Q$ can be expanded by $l(\leq k)$-support Pauli strings as follows:
\eq{
Q = \sum_{l=0}^k \sum_{i=1}^N \sum_{\boldA^l \in G^l} q_{\boldA^l_i} \boldA^l_i ~,
\label{input}
}
where $\boldA^l_i$ are Pauli string $\boldA$ starting from site $i$.
Note that all coefficients $\{ q_{\boldA_i^l} \}$ are real numbers.

Similarly, we expand $[Q,H]$ in the same basis. 
An important observation is that the commutator of two local Pauli strings is another local Pauli string with some coefficient, unless they are commutative.
Specifically, if a $k_1$-support operator and a $k_2$-support operator are non-commutative,
their commutator is a $k_3$-support operator with 
\eq{
|k_1 - k_2| + 1 \leq k_3 \leq k_1 + k_2 - 1 ~.
\label{triangle}
}
Since the Hamiltonian is a $2$-support quantity,
the commutator $[Q,H]$ with a $k$-support quantity $Q$ is an at-most-$(k+1)$-support quantity,
and it can be written as
\eq{
\frac{1}{2i} [ Q, H] = \sum_{l=0}^{k+1} \sum_{i=1}^N \sum_{\boldB^l \in G^l} r_{\boldB^l_i} \boldB^l_i ~,
\label{output}
}
where all coefficients $r_{\boldB^l_i}$ are real numbers due to the division by $2i$.

In the proof of non-integrability, we show that a candidate $k$-support conserved quantity $Q$ satisfies $\qa = 0$ for all $\aki$ from the conservation condition $r_{\boldB^l_i} = 0$ for all $\boldB^l_i$, which leads to the absence of $k$-support conserved quantities.
To grasp how $\qa$ and $r_{\boldB^l_i}$ are connected, let us see an example.
When $Q$ has the term $(XYZ)_i := X_i Y_{i+1} Z_{i+2}$ with a coefficient $q_{(XYZ)_i}$ and $H$ does not have $ZZ$ interaction (i.e., $J_Z=0$), the commutator $\frac{1}{2i}[Q,H]$ includes the following 11 terms:
\begin{widetext}
\eq{
& &\frac{1}{2i}[ q_{(XYZ)_i} (XYZ)_i, H] \nx
&=&  
q_{(XYZ)_i}
(
J_Y  Y_{i-1}Z_iY_{i+1}Z_{i+2}
+J_X  X_iY_{i+1}Y_{i+2}X_{i+3}
-J_Y  X_i Y_{i+1} X_{i+2} Z_{i+3}
-h_X  X_iZ_{i+1}Z_{i+2}
+h_X  X_iY_{i+1}Y_{i+2} \nx
& &\qquad\qquad+h_Y  Z_iY_{i+1}Z_{i+2}
-h_Y  X_iY_{i+1}X_{i+2}
-h_Z  Y_iY_{i+1}Z_{i+2}
+h_Z  X_iX_{i+1}Z_{i+2}
-J_X  Z_{i+1}Z_{i+2}
+J_Y  Z_iZ_{i+2})
~.
\label{demo}
}
\end{widetext}
Other operators in $Q$ may generate the same operator appearing in the right-hand side, e.g., the first term $Y_{i-1}Z_iY_{i+1}Z_{i+2}$, and by summing all such contributions $r_{Y_{i-1}Z_iY_{i+1}Z_{i+2}}$ is calculated.
Hence, the conservation condition $r_{Y_{i-1}Z_iY_{i+1}Z_{i+2}}=0$ provides a relation including $q_{(XYZ)_i}$.
Using this type of relations, we finally derive $q_{(XYZ)_i}=0$.

\bigskip

Below we shall provide the proof of Theorem~\ref{t:main} with dividing the cases according to the rank of $J$ as mentioned in Subsec.~\ref{subsec:pt}.
Since the cases of rank $0$ and rank $1$ have been resolved in the previous study \cite{chiba2024proof},
in this paper we just explain these result briefly.
We provide a detailed proof for the cases of rank $2$ and rank $3$.

\subsection{Proof for the case of rank 0} In this case, the system is a free-spin model, in which the individual spins do not interact with each other.
This fact directly implies the integrability of this model.

\subsection{Proof for the case of rank 1}
Without loss of generality, we assume that the $ZZ$ interaction is the only nonzero interaction.
By applying the global spin rotation in Subsec.~\ref{subsec:pt}, 
the $Y$-field term can be eliminated.
Thus any model can be reduced into
\eq{
H = \sumi \( \jz + \hx + \hz\)
}
with $J_Z \neq 0$.
The integrability/non-integrability of the model has already been classified as follows.
\begin{itemize}
\item If $h_X = 0$, the model is the longitudinal-field Ising model, which is integrable.
\item If $h_X \neq 0$ and $h_Z = 0$, the model is the transverse-field Ising model, which is integrable.
\item If $h_X \neq 0$ and $h_Z \neq 0$, the model is the mixed-field Ising model, which has been proven to be non-integrable in Ref.~\cite{chiba2024proof}.
\end{itemize}

\subsection{Proof for the case of rank 2} 
Without loss of generality, we assume that the $XX$ interaction and the $YY$ interaction are the nonzero interactions, which reads
\eq{
H = \sumi \loq{  \jx  +\jy \\ +\hx +\hy +\hz }
\label{standard_rank2}
}
with $J_X \neq 0$ and $J_Y \neq 0$.
In this setup, we obtain the following contrasting results:
\begin{itemize}
\item If $h_X = h_Y = 0$, the model is the XY model, which is integrable.
\item If $h_X \neq 0$ or $h_Y \neq 0$, the model is proved to be non-integrable,
i.e., it has no $k$-support conserved quantity with $3\leq k \leq N/2$.
\end{itemize}
Since $X$ and $Y$ are treated in the same way in Eq.~\eqref{standard_rank2}, the latter can be interpreted as the following lemma:
\begin{lem}
Any model in Eq.~\eqref{standard_rank2} with $h_X\neq 0$ has no $k$-support conserved quantity with $3\leq k \leq N/2$.
\label{lem1}
\end{lem}

\begin{proof}[Proof of Lemma~\ref{lem1}]
Below we prove that a candidate $k$-support conserved quantity $Q$ must satisfy $\qa = 0$ for all $\aki$ [defined in Eq.~\eqref{input}], considering the conservation condition $r_{\boldB^l_i} = 0$ for all $\boldB^l_i$ [defined in Eq.~\eqref{output}] with all $l$.

In order to visually grasp commutation relations, we introduce the \textit{column expression} of commutators.
We first present an example and then explain the rule of this expression.
The following commutation relation, for example,
\eq{
[X_i Y_{i+1} Z_{i+2}, Y_{i-1} Y_i] = +2i \, Y_{i-1} Z_i Y_{i+1} Z_{i+2}~,
}
which corresponds to the first term of the right hand side of \eqref{demo}, is expressed in the column expression as
\eq{
\bmat{
 &        & X_i & Y_{i+1} & Z_{i+2} \\
 &Y_{i-1} & Y_i \\ \hline
+&Y_{i-1} & Z_i & Y_{i+1} & Z_{i+2}
}~.
}
The column expression has three horizontal rows, two arguments of a commutator are in the first and the second rows, and the resulting operator is in the last row.
A horizontal line separates the two inputs of a commutator and the outputs of it.
The horizontal positions of these three operators reflect the spatial position of the sites on which they act.
When considering a term of $[Q,H]$, we express a term of $Q$ as the first row and a term of $H$ as the second row.
Coefficient $2i$ shall be omitted.

As can be seen from Eq.~\eqref{demo}, $\frac{1}{2i}[\bullet,H]$ is a linear mapping on the space of local quantities.
Thus, the condition that $[Q,H] = 0$, i.e., $r_{\boldB^l_i} = 0$ for all $\boldB^l_i$ in Eq.~\eqref{output}, gives a system of linear equations of the coefficients $\{q_{\boldA^l_i}\}$.

We prove Lemma~\ref{lem1} by showing that the system of linear equations has no nontrivial solution.
In particular, we will show that a candidate $k$-support conserved quantity $Q$ satisfies $\qa=0$ for all $\boldA^k$ and $i$.
The proof has two steps: In Step 1, we focus on the conditions $\{q_{\boldB_j^{k+1}} = 0\}$, and derive that most of $\aki$ have zero coefficients.
In Step 2, we turn to the conditions $\{ q_{\boldB_j^k} = 0\}$, and derive that the remaining $\aki$ also have zero coefficients.

\bigskip

\paragraph*{Step 1---} 
Using the condition that $\rB = 0$ for all $\Bki$, 
we here narrow down candidates of $\aki$ whose coefficients $\qa$ may have non-zero values.
As seen from Eq.~\eqref{triangle},
a $(k+1)$-support operator $\Bki$ can be generated only by $\{ \ak_j \}$ of $Q$ and interaction terms of $H$.
Using this fact, we can obtain simple relations of $\{ \ak_j \}$.

First, we treat the case where both ends of $\ak$ are $Z$ ($\ak = Z \stard Z$). 
To treat this, we consider $\Bki = Z_i \stard Y_{i+k-1} X_{i+k}$ term generated by the commutator $[\aki,H]$, whose column expression is
\eq{
\bmat{
 & Z_i \stara Z_{i+k-1} \\
 &     \starn X_{i+k-1} & X_{i+k}\\ \hline 
+& Z_i \stara Y_{i+k-1} & X_{i+k}
\label{ZZdemo}
}~.}
We see that other terms in $[Q,H]$ cannot contribute to $\Bki = Z_i \stard Y_{i+k-1}  X_{i+k}$, which is confirmed as follows: 
Since $\Bki = Z_i \stard Y_{i+k-1}  X_{i+k}$ is a $(k+1)$-support operator, this operator appears in commutator $[Q,H]$ only in the case that the commutator consists of a $k$-support operator in $Q$ and a 2-support operator in $H$ and these two operators have an overlap only on a single site.
Hence, the following two forms of commutators are the only candidates that give rise to $\Bki = Z_i \stard Y_{i+k-1}  X_{i+k}$:
\eq{
\bmat{
  & ? & \cdots &  ? & ? & ? \\
 ?&? \\ \hline 
 Z_i \stara Y_{i+k-1} & X_{i+k}}
\qquad 
\bmat{
 ? & ? & \cdots & ? & ? \\
   &   &        &   & ? & ?\\ \hline 
Z_i \stara Y_{i+k-1} & X_{i+k}
}~. \nx}
Here, we further observe that the left commutator does not exist because $H$ has no 2-support term whose first operator is $Z$.
In addition, the right form of a commutator exists only that in Eq.~\eqref{ZZdemo} because a 2-support term in $H$ whose last operator is $X$ is uniquely determined as $XX$.
By the above discussion, we conclude that the only term in $[Q,H]$ that contributes to $\Bki = Z_i\stard Y_{i+k-1}X_{i+k}$ is $[Z_i\stard Z_{i+k-1},X_{i+k-1}X_{i+k}]$.
Thus we obtain
\eq{
J_X q_{Z_i \stard Z_{i+k-1}} = r_{Z_i \stard Y_{i+k-1} X_{i+k}} = 0
\label{immediately0}
~,}
which means that all operators whose first and last operator is $Z$ (i.e., $\ak = Z \stard Z$) has zero coefficient in $Q$.

In a similar discussion as above, we can show that all operators in $Q$ whose one end is $Z$ have zero coefficients.
Owing to the inversion symmetry, the remaining cases are $\ak = Z\stard X$ and $Z\stard Y$.
We first treat $\ak = Z \stard X$. 
Here, we consider commutators generating $\Bki = (Z \stard Z Y)_i := Z_i \stard Z_{i+k-1} Y_{i+k}$.
We find that the only contribution from $[Q,H]$ to $\Bki = (Z \stard Z Y)_i$ is
\eq{
\bmat{
 & Z_i \stara X \\
 &     \starn Y & Y \\ \hline
+& Z_i \stara Z & Y
}~
}
and no other terms of $[Q,H]$ contribute to $\Bki = (Z \stard Z Y)_i$.
As a result, $\ak = Z \stard X$ is shown to have zero coefficient.
Similarly,
\eq{
\bmat{
 & Z_i \stara Y \\
 &     \starn X & X \\ \hline
-& Z_i \stara Z & X
}
}
is the only contribution to $\Bki = (Z \stard ZX)_i$, which suggests that $\ak = Z \stard Y$ also has zero coefficient.
In summary, we conclude $q_{(Z \stard)_i} = 0$.
That is, $\qa$ may take non-zero values only if the left end of $\ak$ is $X$ or $Y$.

We shall go forward to our analysis on the second left operator of $\ak$.
We first treat $\ak = X X \stard $.
Here we notice that
\eq{
\bmat{
 & X_i & X \stara X \\
 &     &   \starn Y & Y \\ \hline
+& X_i & X \stara Z & Y
}
}
is the only contribution to $\Bki = (X X \stard Z Y)_i$, which is confirmed as follows:
The other candidate of commutators which may contribute to $\Bki = (X X \stard Z Y)_i$ takes the form of
\eq{
\bmat{
 &     & ? & ? & \cdots & ? & ? & ? \\
 & X   & X  \\ \hline
 & X_i & X \stara Z & Y
}~.
}
However, no operator on site $i+1$ satisfies this commutation relation.
Note that even if the rightmost operator is $Y$ or $Z$, we can make similar arguments by considering commutator with appropriate interaction terms.
Thus we obtain $q_{(XX \stard)_i} = 0$.

For a similar reason, in the case of $\ak = X I \stard $, we notice that
\eq{
\bmat{
 & X_i & I \stara X \\
 &     &   \starn Y & Y \\ \hline
+& X_i & I \stara Z & Y
}
}
is the only contribution to $\Bki = (X I \stard Z Y)_i$, which leads to $q_{(XI\stard)_i} = 0$.

Similar arguments hold for the case that the leftmost operator is $Y$, and we find that $\aki = (YY \stard)_i$ and $(YI \stard)_i$ have zero coefficient.
The remaining operators which may have nonzero coefficient take the form of $\aki = (XY \stard)_i$, $(XZ \stard)_i$, $(YX \stard)_i$, and $(YZ \stard)_i$.
Below, we demonstrate that $\aki = (XY \stard)_i$ (and $(YX \stard)_i$) has zero coefficient.

To this end, we consider commutators which contribute to $\Bki = (XY\stard ZY)_i$.
We notice that the following to commutators
\eq{
\bmat{
 & X_i & Y \stara X \\
 &     &   \starn Y & Y \\ \hline
+& X_i & Y \stara Z & Y
}
\qquad
\bmat{
 &     & Z \stara Z & Y \\
 & X   & X       \starn       \\ \hline
+& X_i & Y       \stara Z & Y
}~, \nx
}
contribute to $\Bki = (XY\stard ZY)_i$, which leads to the following linear relation of coefficients:
\eq{
J_Y q_{(X Y \stard X)_i} + J_X q_{(Z \stard ZY)_{i+1}} &&= r_{(XY \stard ZY)_i} \nx &&= 0 ~.
}
Recalling that the second term of the left-hand side $q_{(Z\stard ZY)_{i+1}}$ has already been shown to be zero,
we obtain $q_{(X Y \stard X)_i}=0$.
Similarly, we obtain $q_{(YX \stard)_i} = 0$.
In summary, $\qa = 0$ holds except for $\ak = XZ \stard$ and $YZ \stard$.

We shall go forward to our analysis on the third left operator of $\ak$.
Commutators which contribute to $\bk = (XZY \stard ZY)_i$ are written as
\begin{widetext}
\eq{
\bmat{
 & X_i & Z & Y \stara X \\
 &     &   &   \starn Y & Y \\ \hline
+& X_i & Z & Y \stara Z & Y
}
\qquad
\bmat{
 &     & Y & Y \stara Z & Y \\
 & X   & X &     \starn       \\ \hline
-& X_i & Z & Y       \stara Z & Y
}~, \nx
}
\end{widetext}
which leads to
\eq{
J_Y q_{(X Z Y \stard X)_i} - J_X q_{(Y Y \stard ZY)_{i+1}} = 0 ~.
}
Since $q_{(YY \stard ZY)_{i+1}} = 0$ holds, we obtain $q_{(X Z X \stard X)_i} = 0$.
Similar arguments yield
\eq{
q_{(XZI \stard X)_i} &\propto& q_{(YI \stard ZY)_{i+1}} = 0 ~,\nx
q_{(XZX \stard X)_i} &\propto& q_{(YX \stard ZY)_{i+1}} \propto q_{(Z\stard ZZX)_{i+2}} = 0 ~. \nx
\label{proportional0}
}
As a consequence, $\qa = 0$ holds except for $\ak = XZZ \stard$ and $YZZ \stard$.

By repeating similar arguments for the fourth left operator onwards, we find that $\qa = 0$ holds for all $i$ except for $\ak$ such that the both ends are $X$ or $Y$ and all the other middle operators are $Z$, i.e., $\ak = X(Z)^{k-2}X,X(Z)^{k-2}Y,Y(Z)^{k-2}X$, or $Y(Z)^{k-2}Y$.

In addition, the coefficients of the above remaining candidates satisfy the following relations:
\eq{
\frac{q_{(X(Z)^{k-2}X)_i}}{J_X} &=& \frac{q_{(Y(Z)^{k-2}Y)_{i+1}}}{J_Y} = \frac{q_{(X(Z)^{k-2}X)_{i+2}}}{J_X} = \cdots ~, \nx
q_{(X(Z)^{k-2}Y)_i} &=& -q_{(Y(Z)^{k-2}X)_{i+1}} = q_{(X(Z)^{k-2}Y)_{i+2}} = \cdots ~. \nx
\label{linear_rank2}
}
The first equality in the upper row, for example, is confirmed by considering
\eq{
\bmat{
 & X_i & (Z)^{k-2} & X \\
 &     &           & Y & Y \\ \hline
+& X_i & (Z)^{k-2} & Z & Y
}
\qquad
\bmat{
 &     & Y_{i+1} & (Z)^{k-2} & Y \\
 & X   & X                       \\ \hline
-& X_i & Z       & (Z)^{k-2} & Y
}~. \nx
}
Similar arguments confirms the other relations.

We summarize our findings through Step 1 as follows: (i) If $\ak \notin \{ X(Z)^{k-2}X,X(Z)^{k-2}Y,Y(Z)^{k-2}X,Y(Z)^{k-2}Y \}$, then $\qa = 0$ for all $i$. (ii) The remaining coefficients are not independent but connected through linear relations as in Eq.~\eqref{linear_rank2}.
Owing to these linear relations, there are at most four independent coefficients; the coefficients of $(X(Z)^{k-2}X)_{i=1}$, $(X(Z)^{k-2}X)_{i=2}$, $(Y(Z)^{k-2}X)_{i=1}$, and $(Y(Z)^{k-2}X)_{i=2}$.
Note that if the system size $N$ is odd, the coefficients of the first and the third operators are the only independent coefficients, and those of the second and the fourth are no longer independent.

\bigskip

\paragraph*{Step 2 ---} 
We next derive $\qa = 0$ for the remaining candidates $\ak$ from the conditions on $k$-support operators, i.e., $\rb = 0$ for all $\bki$.
As seen from Eq.~\eqref{triangle}, only $k$-support operators $\{ \ak_j \}$ and $(k-1)$-support operators $\{ \boldA^{k-1}_j \}$ may contribute to $\bki$.
Thanks to Eq.~\eqref{linear_rank2}, it suffices to show  $q_{(X(Z)^{k-2}X)_i} = 0$ and $q_{(X(Z)^{k-2}Y)_i} = 0$ for general $i$, which includes the case of $i=1$ and $i=2$.

We take $\bk_{i-1} = (YY(Z)^{k-3}Y)_{i-1}$ and analyze what commutators generate this $\bk_{i-1}$.
All the commutators that contribute to this $\bk_{i-1}$ are
\eq{
\bmat{
 & Y_{i-1} & Z & (Z)^{k-3} & Y \\
 &         & X &           &   \\ \hline
+& Y_{i-1} & Y & (Z)^{k-3} & Y
}
\qquad
\bmat{
 & Y_{i-1} & Y & (Z)^{k-4} & X &   \\
 &         &   &           & Y & Y \\ \hline
+& Y_{i-1} & Y & (Z)^{k-4} & Z & Y
}~,\nx
}
which gives the following equation:
\eq{
  h_X q_{(Y(Z)^{k-2}Y)_{i-1}} 
+ J_Y q_{(YY(Z)^{k-4}X)_{i-1}} = 0 ~.
\label{step2_1}
}
Here, we omit commutators with operators eliminated in Step 1.
Observe that $(Y(Z)^{k-2}Y)_{i-1}$ in the first term is a $k$-support operator and that $(YY(Z)^{k-4}X)_{i-1}$ in the second term is a $(k-1)$-support operator. 

Let us investigate other commutators which contains $\ak_{i-1} = (YY(Z)^{k-4}X)_{i-1}$ in the top row. 
Commutators generating $\bk_{i-2} = (XZY(Z)^{k-4}X)_{i-2}$ satisfy this requirement, which are expressed as
\begin{widetext}
\eq{
&&\bmat{
 & X_{i-2} & Z & Z & (Z)^{k-4} & X  \\
 &         &   & X &           &    \\ \hline
+& X_{i-2} & Z & Y & (Z)^{k-4} & X
}
\qquad
\bmat{
 &         & Y_{i-1} & Y & (Z)^{k-4} & X  \\
 & X       & X       &   &           &    \\ \hline
-& X_{i-2} & Z       & Y & (Z)^{k-4} & X
}
\qquad
\bmat{
 & X_{i-2} & Z & Y & (Z)^{k-5} & Y &   \\
 &         &   &   &           & X & X \\ \hline
-& X_{i-2} & Z & Y & (Z)^{k-5} & Z & X
}~.
}
These commutators yield
\eq{
 h_X q_{(X(Z)^{k-2}X)_{i-2}}
-J_X q_{(YY(Z)^{k-4}X)_{i-1}}
-J_X q_{(XZY(Z)^{k-5}Y)_{i-2}}
=0 ~.
\label{step2_2}
}
By adding Eqs.~\eqref{step2_1} and \eqref{step2_2} with appropriate coefficients, 
we can eliminate $q_{(YY(Z)^{k-4}X)_{i-1}}$.

\bigskip

In Step 2, by adding equations obtained from $\{ r_{\bk_j} = 0 \}$, we eliminate coefficients of $(k-1)$-support operators as above.
The goal of Step 2 is to obtain a nontrivial equation which consists only of $\{ q_{\ak_j} \}$ without $\{ q_{\boldA^{k-1}_j} \}$. 
We note that employed commutators in our proof are slightly different between the cases of even $k$ and odd $k$.

First we treat the case of odd $k$.
For $1 \leq m \leq k-2$, we exploit the conditions $r_{\bk_{i-m}} = 0$ with
\eq{
\bk_{i-m} = \left \{
\begin{array}{ll}
    (Y(Z)^{m-1}Y(Z)^{k-2-m}Y)_{i-m} & {\rm if \,} m {\rm \, is \, odd} \\
    (X(Z)^{m-1}X(Z)^{k-2-m}X)_{i-m} & {\rm if \,} m {\rm \, is \, even} 
\end{array} \right. ~.
}
Considering commutators generating these $\bk_{i-m}$'s, we obtain $k-2$ relations of coefficients $\{q_{\ak_j}\}$ and $\{ q_{\boldA^{k-1}_j} \}$, which reads
\eq{
\begin{array}{llllll}
& h_X q_{(Y(Z)^{k-2}Y)_{i-1}}
&
&+J_Y q_{(YY(Z)^{k-4}X)_{i-1}}
&=0
&(m=1)\\
& h_X q_{(X(Z)^{k-2}X)_{i-2}}
&-J_X q_{(YY(Z)^{k-4}X)_{i-1}}
&-J_X q_{(XZY(Z)^{k-5}Y)_{i-2}}
&=0
&(m=2)\\
& h_X q_{(Y(Z)^{k-2}Y)_{i-3}}
&+J_Y q_{(XZY(Z)^{k-5}Y)_{i-2}}
&+J_Y q_{(YZZY(Z)^{k-6}X)_{i-3}}
&=0
&(m=3)\\
&&& \vdots && \\
& h_X q_{(X(Z)^{k-2}X)_{i+3-k}}
&-J_X q_{(Y(Z)^{k-5}YZX)_{i+4-k}}
&-J_X q_{(X(Z)^{k-4}YY)_{i+3-k}}
&=0
&(m=k-3)\\
& h_X q_{(Y(Z)^{k-2}Y)_{i+2-k}}
&+J_Y q_{(X(Z)^{k-4}YY)_{i+3-k}}
&
&=0
&(m=k-2)~.\\
\end{array}
}
Adding these equations together with multiplying appropriate factors (concretely, multiplying $J_X/J_Y$ if $m$ is odd, and $1$ if $m$ is even) yields
\eq{
(k-2) h_X q_{ (X(Z)^{k-2}X)_i } = 0 ~.
\label{zero1_rank2}
}
Since $h_X \neq 0$ was assumed, we obtain $q_{ (X(Z)^{k-2}X)_i } = 0$ for all $i$.

We next show that $q_{ (X(Z)^{k-2}X)_i }$ and $q_{(Y(Z)^{k-2}X)_i}$ are connected by a linear relation.
For $0 \leq m \leq k-2$, we exploit the conditions $r_{\bk_{i-m}} = 0$ with
\eq{
\bk_{i-m} = \left \{
\begin{array}{ll}
    ((Z)^{k-1}X)_{i-m} & {\rm if \,} m = 0 \\
    (X(Z)^{m-1}Y(Z)^{k-2-m}Y)_{i-m} & {\rm if \,} m {\rm \, is \, odd} \\
    (Y(Z)^{m-1}X(Z)^{k-2-m}X)_{i-m} & {\rm if \,} m {\rm \, is \, postive \, even}
\end{array} \right. ~.
}
{Considering commutators generating these $\bk_{i-m}$'s, we obtain equations of $\{q_{\ak_j}\}$ and $\{ q_{\boldA^{k-1}_j} \}$ as}
\eq{
\begin{array}{llllll}
&- h_X q_{(Y(Z)^{k-2}X)_i}+h_Y q_{(X(Z)^{k-2}X)_i}
&
&-J_X q_{((Z)^{k-2}Y)_i}
&=0
&(m=0)\\
& h_X q_{(X(Z)^{k-2}Y)_{i-1}}
&+J_X q_{((Z)^{k-2}Y)_i}
&+J_Y q_{(XY(Z)^{k-4}X)_{i-1}}
&=0
&(m=1)\\
& h_X q_{(Y(Z)^{k-2}X)_{i-2}}
&+J_Y q_{(XY(Z)^{k-4}X)_{i-1}}
&-J_X q_{(YZY(Z)^{k-5}Y)_{i-2}}
&=0
&(m=2)\\
&&& \vdots && \\
& h_X q_{(Y(Z)^{k-2}X)_{i+3-k}}
&+J_Y q_{(X(Z)^{k-5}YZX)_{i+4-k}}
&-J_X q_{(Y(Z)^{k-4}YY)_{i+3-k}}
&=0
&(m=k-3)\\
& h_X q_{(X(Z)^{k-2}Y)_{i+2-k}}
&-J_X q_{(Y(Z)^{k-4}YY)_{i+3-k}}
&
&=0
&(m=k-2)~.\\
\end{array}
}
Adding these equations together with multiplying appropriate factors (concretely, $J_X/J_Y$ if $m = 0$ or odd, and $1$ if $m$ is positive even) yields
\eq{
(k-1) h_X q_{(Y(Z)^{k-2}X)_i} - h_Y q_{(X(Z)^{k-2}X)_i} = 0
\label{zero2_rank2} ~.
}
{Since $q_{(X(Z)^{k-2}X)_i}$ in the second term is shown to be zero just before and $h_X$ is assumed to be nonzero,
we obtain $q_{ (Y(Z)^{k-2}X)_i } = 0$ for all $i$.
Consequently, all four independent coefficients are shown to be zero.
Step 2 has been completed for the case of odd $k$.}

\bigskip

We next give the equations corresponding to Eqs.~\eqref{zero1_rank2} and \eqref{zero2_rank2} for the case $k$ is even.
Below we make a parallel argument to the case where $k$ is odd.

For $1 \leq m \leq k-2$, we exploit $r_{\bk_{i-m}} = 0$ with
\eq{
\bk_{i-m} = \left \{
\begin{array}{ll}
    (Y(Z)^{m-1}Y(Z)^{k-2-m}X)_{i-m} & {\rm if \,} m {\rm \, is \, odd} \\
    (X(Z)^{m-1}X(Z)^{k-2-m}Y)_{i-m} & {\rm if \,} m {\rm \, is \, even}
\end{array} \right. ~.
}
{Considering commutators generating these $\bk_{i-m}$'s, we obtain equations of $\{q_{\ak_j}\}$ and $\{ q_{\boldA^{k-1}_j} \}$ as}
\eq{
\begin{array}{llllll}
& h_X q_{(Y(Z)^{k-2}X)_{i-1}}
&
&-J_X q_{(YY(Z)^{k-4}Y)_{i-1}}
&=0
&(m=1)\\
& h_X q_{(X(Z)^{k-2}Y)_{i-2}}
&-J_X q_{(YY(Z)^{k-4}Y)_{i-1}}
&+J_Y q_{(XZY(Z)^{k-5}X)_{i-2}}
&=0
&(m=2)\\
& h_X q_{(Y(Z)^{k-2}X)_{i-3}}
&+J_Y q_{(XZY(Z)^{k-5}X)_{i-2}}
&-J_X q_{(YZZY(Z)^{k-6}Y)_{i-3}}
&=0
&(m=3)\\
&&& \vdots && \\
& h_X q_{(Y(Z)^{k-2}X)_{i+3-k}}
&+J_Y q_{(X(Z)^{k-5}YZX)_{i+4-k}}
&-J_X q_{(Y(Z)^{k-4}YY)_{i+3-k}}
&=0
&(m=k-3)\\
& h_X q_{(X(Z)^{k-2}Y)_{i+2-k}}
&-J_X q_{(Y(Z)^{k-4}YY)_{i+3-k}}
&
&=0
&(m=k-2)~.\\
\end{array}
}
Adding these equations together with multiplying appropriate factors (concretely, $-1$ if $m$ is odd, and $1$ if $m$ is even) yields
\eq{
(k-2) h_X q_{ (X(Z)^{k-2}Y)_i } = 0 ~.
\label{zero3_rank2}
}
Since $h_X \neq 0$ was assumed, we obtain $q_{ (X(Z)^{k-2}Y)_i } = 0$ for all $i$.

Next, we shall show that the remaining coefficient $q_{(Y(Z)^{k-2}Y)_i}$ is zero.
For $0 \leq m \leq k-2$, we exploit the conditions $r_{\bk_{i-m}} = 0$ with
\eq{
\bk_{i-m} = \left \{
\begin{array}{ll}
    ((Z)^{k-1}Y)_{i-m} & {\rm if \,} m = 0 \\
    (X(Z)^{m-1}Y(Z)^{k-2-m}X)_{i-m} & {\rm if \,} m {\rm \, is \, odd} \\
    (Y(Z)^{m-1}X(Z)^{k-2-m}Y)_{i-m} & {\rm if \,} m {\rm \, is \, postive \, even}
\end{array} \right. ~.
}
{Considering equations of $\{q_{\ak_j}\}$ and $\{ q_{\boldA^{k-1}_j} \}$ which generate these $\bk_{i-m}$'s, we obtain}
\eq{
\begin{array}{llllll}
&- h_X q_{(Y(Z)^{k-2}Y)_i}+h_Y q_{(X(Z)^{k-2}Y)_i}
&
&+J_Y q_{((Z)^{k-2}X)_i}
&=0
&(m=0)\\
& h_X q_{(X(Z)^{k-2}X)_{i-1}}
&+J_X q_{((Z)^{k-2}X)_i}
&-J_X q_{(XY(Z)^{k-4}Y)_{i-1}}
&=0
&(m=1)\\
& h_X q_{(Y(Z)^{k-2}Y)_{i-2}}
&+J_Y q_{(XY(Z)^{k-4}Y)_{i-1}}
&+J_Y q_{(YZY(Z)^{k-5}X)_{i-2}}
&=0
&(m=2)\\
&&& \vdots && \\
& h_X q_{(X(Z)^{k-2}X)_{i+3-k}}
&-J_X q_{(Y(Z)^{k-5}YZX)_{i+4-k}}
&-J_X q_{(X(Z)^{k-4}YY)_{i+3-k}}
&=0
&(m=k-3)\\
& h_X q_{(Y(Z)^{k-2}Y)_{i+2-k}}
&+J_Y q_{(X(Z)^{k-4}YY)_{i+3-k}}
&
&=0
&(m=k-2)~.\\
\end{array}
}
\end{widetext}
Adding these equations together with multiplying appropriate factors (concretely, $-1$ if $m = 0$, $J_Y/J_X$ if $m$ is odd, and $1$ if $m$ is positive even) yields
\eq{
(k-1) h_X q_{(Y(Z)^{k-2}Y)_i} - h_Y q_{(X(Z)^{k-2}Y)_i} = 0 ~.
\label{zero4_rank2}
}
Since the second term on the left-hand side is zero and $h_X \neq 0$, we obtain $q_{ (Y(Z)^{k-2}Y)_i } = 0$ for all $i$.
Note that all the above discussion in Step 2 is independent of the values of $h_Y$ and $h_Z$.

Combining Step 1 and Eqs.~\eqref{zero1_rank2}, \eqref{zero2_rank2}, \eqref{zero3_rank2} and \eqref{zero4_rank2}, we conclude $\qa = 0$ for all $k$-support operators $\aki$.
This completes the proof of the absence of $k$-support conserved quantity for systems represented by Eq.~\eqref{standard_rank2}.

\end{proof}

\subsection{Proof for the case of rank 3}
The standard form of Hamiltonians in this case is
\eq{
H = \sumi \begin{array}{l}  (\jx  +\jy +\jz \\ +\hx +\hy +\hz) \end{array} ~,
\label{standard_rank3}
}
with $J_X, J_Y, J_Z \neq 0$.
Here we summarize our result of integrability/non-integrability of this model.
\begin{itemize}
\item If $J_X=J_Y=J_Z$ holds, then the system is the XXX model.
This model is integrable, regardless of the values of $h_X, h_Y$ and $h_Z$.
\item If $J_X = J_Y \neq J_Z$ holds,
\begin{itemize}
\item if $h_X = h_Y = 0$, the system is the XXZ model.
This model is integrable.
\item if $h_X \neq 0$ or $h_Y \neq 0$, the systems is non-integrable, as will be shown in Lemma~\ref{lem2}.
\end{itemize}
If $J_Y = J_Z \neq J_X$ or $J_Z = J_X \neq J_Y$ hold, the system can be reduced to the case of $J_X = J_Y \neq J_Z$ by a global spin rotation.
\item If $J_X,J_Y,$ and $J_Z$ are all different,
\begin{itemize}
\item if $h_X = h_Y = h_Z = 0$, the system is the XYZ model.
This model is integrable.
\item if $h_X \neq 0$ or $h_Y \neq 0$ or $h_Z \neq 0$, the system is non-integrable as will be shown in Lemma~\ref{lem2}.
\end{itemize}
\end{itemize}

The non-integrability of the two cases listed above follows immediately from the following lemma.

\begin{lem}
Any model in Eq.~\eqref{standard_rank3} with $J_X \neq J_Y$ and $h_Z \neq 0$ has no $k$-support conserved quantity with $3 \leq k \leq N/2$.
\label{lem2}
\end{lem}

\begin{proof}[Proof of Lemma~\ref{lem2}]
The basic strategy is the same as that for the proof of Lemma~\ref{lem1}.
We show that a candidate $k$-support conserved quantity $Q$ satisfies $\qa = 0$ for all $\aki$,
using $\rB = 0$ for all $\Bki$ in Step 1 and $\rb = 0$ for all $\bki$ in Step 2.
To explain in detail, in Step 1, we obtain that coefficients of most $k$-support operators are zero, except a specific form of operators called doubling-product operators.
We also obtain that all the remaining coefficients of doubling-product operators are in the linear relation.
In Step 2, we show that one of the remaining coefficient is zero, which leads to the fact that all coefficients are zero.

\bigskip
\paragraph*{Step 1 ---} 
Since this discussion is exactly the same as Step 1 in \cite{shiraishi2019proof}, we here just describe the results and omit its proof.

For convenience,
we abbreviate $P_i P_{i+1}$ as $\overline{P}_i$ for any $P \in \{X,Y,Z\}$,
and introduce the product of $\overline{P}_i$ on consecutive sites.
We promise that a neighboring $\overline{P}$ sits on sites with one-site shift.
For example, $(\overline{XYX})_i$ represents
\eq{
(\overline{XYX})_i 
&=& c \overline{X}_i \overline{Y}_{i+1} \overline{X}_{i+2} \nx
&=& c (X_i X_{i+1}) (Y_{i+1} Y_{i+2}) (X_{i+2} X_{i+3}) \nx
&=& X_i Z_{i+1} Z_{i+2} X_{i+3} \nx
&=& (XZZX)_i ~,
}
where $c$ is one of $\{ \pm 1, \pm i \}$ to normalize the phase factor to $+1$.
We promise that $\overline{P}_j$ and $\overline{P}_{j+1}$ with the same Pauli operator $P$ do not neighbor.
We call operators which can be expressed in the above form as \textit{doubling-product operators}.
Examples of doubling-product operators are $(X Y Z)_i = (\overline{XZ})_i$ and $(ZYYZ)_i = (\overline{ZXZ})_i$.

In Ref.~\cite{shiraishi2019proof}, through the argument of Step 1 the following two statements are proven: (i) If $\aki$ is not a doubling-product operator, then $\qa=0$.
(ii) Each coefficient $\qa$ of a doubling-product operator is proportional to other remaining coefficients.
The proof strategy is similar to Step 1 of the rank 2 case.
Some $(k+1)$-support operator $\bk$ are generated only by a single commutator with $k$-support operator in $Q$ and 2-support operator in $H$, which leads to the fact that the coefficients $\qa$ of operators in $Q$ appearing in the above commutators are immediately shown to be zero as Eq.~\eqref{immediately0}.
In addition, some other $(k+1)$-support operator $\bki$ are generated by two commutators with $k$-support operators in $Q$ and 2-support operators in $H$, which suggests linear relations of coefficients.
Using these linear relations repeatedly, we find that most of the coefficients of $k$-support operators are shown to be zero as Eq.~\eqref{proportional0}, and the remaining coefficients are all connected by linear relations.
Readers who are interested in the details of its proof are invited to Ref.~\cite{shiraishi2024absence}, which provide pedagogical explanation of the proof.

\bigskip
\paragraph*{Step 2 ---} 
We next derive $\qa = 0$ for the remaining candidates $\aki$ from $\rb = 0$ for all $\bki$.

Depending on whether $k$ is odd or even, we focus on different $\bk_{i}$. When $k$ is odd,  we exploit the conditions $r_{\bk_{i-m}} = 0$ for $0 \leq m \leq k-2$ with
\eq{
\bk_{i-m} = \left \{
\begin{array}{ll}
    ((X)^{k-1}Z)_{i-m} & {\rm if \,} m = 0 \\
    (Z(X)^{m-1}Y(X)^{k-2-m}Y)_{i-m} & {\rm if \,} m {\rm \, is \, odd} \\
    (Y(X)^{m-1}Y(X)^{k-2-m}Z)_{i-m} & {\rm if \,} m {\rm \, is \, postive \, even}
\end{array} \right. ~.
}
Considering commutators generating these $\bk_{i-m}$'s, we obtain the following equations of $\{ q_{\ak_j} \}$ and $\{ q_{\boldA^{k-1}_j} \}$:
\begin{widetext}
\eq{
\! \! \! \!
\begin{array}{lllllll}
& h_Z q_{(Y(X)^{k-2}Z)_i}
&+h_Z q_{(XY(X)^{k-3}Z)_i}
&
&+J_Z q_{((X)^{k-2}Y)_i}
&=0
&(m=0)\\
&- h_Z q_{(Z(X)^{k-2}Y)_{i-1}}
&+h_Z q_{(ZYY(X)^{k-4}Y)_{i-1}}
&-J_Z q_{((X)^{k-2}Y)_{i}}
&-J_Y q_{(ZY(X)^{k-5}Y)_{i-1}}
&=0
&(m=1)\\
&- h_Z q_{(Y(X)^{k-2}Z)_{i-2}}
&+h_Z q_{(YXYY(X)^{k-5}Z)_{i-2}}
&-J_Y q_{(ZY(Z)^{k-4}X)_{i-1}}
&+J_Z q_{(YXY(Z)^{k-5}Y)_{i-2}}
&=0
&(m=2)\\
&&& \vdots && \\
&- h_Z q_{(Y(X)^{k-2}Z)_{i+3-k}}
&+h_Z q_{(Y(X)^{k-4}YYZ)_{i+3-k}}
&-J_Y q_{(Z(X)^{k-5}YXZ)_{i+4-k}}
&+J_Z q_{(Y(X)^{k-4}YY)_{i+3-k}}
&=0
&(m=k-3)\\
&- h_Z q_{(Z(X)^{k-2}Y)_{i+2-k}}
&- h_Z q_{(Z(X)^{k-3}YX)_{i+2-k}}
&-J_Z q_{(Y(X)^{k-4}YY)_{i+3-k}}
&
&=0
&(m=k-2)~.
\end{array} \nx
}

Note that $h_X$ and $h_Y$ do not appear in the equations, since $\frac{1}{2i}[Q,\sumi \hx]$ and $ \frac{1}{2i}[Q,\sumi \hy]$ do not contribute to these $\bk_{i-m}$'s, under the situation where $\aki$'s are doubling-product operators.
Adding these equations together with multiplying appropriate factors (concretely, we multiply $1$ if $m = 0$ or $m$ is odd, and multiply $-1$ if $m$ is positive even) yields
\eq{
(k-1) h_Z \( 1- \frac{J_X}{J_Y} \) q_{ (Y(X)^{k-2}Z)_i } = 0 ~,
}
where we used the linear relations between the coefficients of the $k$-support doubling-product operators obtained in Step~1.
Since we assume $h_Z\neq 0$ and $J_X\neq J_Y$, the above relation means that the coefficient of a $k$-support doubling-product operator is zero, which leads to the desired result that all the $k$-support doubling-product operators have zero coefficients for odd $k$.

When $k$ is even, we exploit the conditions $r_{\bk_{i-m}} = 0$ for $0 \leq m \leq k-2$ with
\eq{
\bk_{i-m} = \left \{
\begin{array}{ll}
    ((X)^{k-1}Y)_{i-m} & {\rm if \,} m = 0 \\
    (Z(X)^{m-1}Y(X)^{k-2-m}Z)_{i-m} & {\rm if \,} m {\rm \, is \, odd} \\
    (Y(X)^{m-1}Y(X)^{k-2-m}Y)_{i-m} & {\rm if \,} m {\rm \, is \, postive \, even}
\end{array} \right. ~.
}
Considering commutators generating these $\bk_{i-m}$'s, we obtain equations of $\{ q_{\ak_j} \}$ and $\{ q_{\boldA^{k-1}_j} \}$ as
\eq{
\! \! \! \!
\begin{array}{lllllll}
& h_Z q_{(Y(X)^{k-2}Y)_i}
&+h_Z q_{(XY(X)^{k-3}Y)_i}
&
&-J_Y q_{((X)^{k-2}Z)_i}
&=0
&(m=0)\\
&- h_Z q_{(Z(X)^{k-2}Z)_{i-1}}
&+h_Z q_{(ZYY(X)^{k-4}Z)_{i-1}}
&-J_Z q_{((X)^{k-2}Z)_{i}}
&+J_Z q_{(ZY(X)^{k-5}Y)_{i-1}}
&=0
&(m=1)\\
&- h_Z q_{(Y(X)^{k-2}Y)_{i-2}}
&+h_Z q_{(YXYY(X)^{k-5}Y)_{i-2}}
&-J_Y q_{(ZY(Z)^{k-4}Y)_{i-1}}
&+J_Y q_{(YXY(Z)^{k-5}Z)_{i-2}}
&=0
&(m=2)\\
&&& \vdots && \\
&- h_Z q_{(Z(X)^{k-2}Z)_{i+3-k}}
&+h_Z q_{(Z(X)^{k-4}YYZ)_{i+3-k}}
&-J_Z q_{(Y(X)^{k-5}YXZ)_{i+4-k}}
&+J_Z q_{(Z(X)^{k-4}YY)_{i+3-k}}
&=0
&(m=k-3)\\
&- h_Z q_{(Y(X)^{k-2}Y)_{i+2-k}}
&- h_Z q_{(Y(X)^{k-3}YX)_{i+2-k}}
&-J_Y q_{(Z(X)^{k-4}YY)_{i+3-k}}
&
&=0
&(m=k-2)\\
\end{array} ~.
}
\end{widetext}
Adding these equations together with multiplying appropriate factors (concretely, $1$ if $m = 0$, $-J_Y/J_Z$ if $m$ is odd, and $-1$ if $m$ is positive even) yields
\eq{
(k-1) h_Z \( 1- \frac{J_X}{J_Y} \) q_{ (Y(X)^{k-2}Y)_i} = 0 ~,
}
where we used the linear relations between the coefficients of the $k$-support doubling-product operators obtained in Step~1.
Since we assume $h_Z\neq 0$ and $J_X\neq J_Y$, the above relation means that the coefficient of a $k$-support doubling-product operator is zero, which leads to the desired result that all the $k$-support doubling-product operators have zero coefficients for even $k$.

Thus, if $h_z \neq 0$ and $ 1 - \frac{J_X}{J_Y}  \neq 0$ holds, we have $\qa = 0$ for one of the $k$-support doubling-product operators,
and combining with the result of Step 1, 
we conclude $\qa = 0$ for all $k$-support operators $\aki$.
This means that the system has no $k$-support conserved quantity (with $3 \leq k \leq N/2$).
\end{proof}

\section{Short-support conserved quantities of non-integrable systems}\label{sec:Trivial}

In the previous sections, we do not care about $k$-support conserved quantities with $k \leq 2$ as they are ``trivial'' local conserved quantities.
Indeed, the mere existence of $k$-support conserved quantities with $k\leq 2$ is  far from sufficient for the solvability of systems. 

However, it is important to consider such short-support conserved quantities of non-integrable systems in the study of thermalization and level statistics.
With these motivations, we investigate $k(\leq 2)$-support conserved quantities of non-integrable systems described in the form of Eq.~\eqref{general}.

Common short-support conserved quantities are the Hamiltonian ($2$-support) and the identity operator ($0$-support).
Although most non-integrable systems described in the form of Eq.~\eqref{general} have only these two local conserved quantities, some systems have additional local conserved quantities.

We first present the result (pairs of a non-integrable Hamiltonian and a short-support conserved quantity) and then prove that no further short-support conserved quantity exists.
There are two pairs of $(H,Q)$:
One is
\eq{
H &=& \sumi \Big( J_X \Big( X_i X_{i+1} -  \Big(\frac{h_Y}{h_X}\Big)^2 Y_i Y_{i+1} \Big) + h_X X_i + h_Y Y_i \Big) ~, \nx
Q &=& \sumi (-1)^i (h_X X_i + h_Y Y_i) (h_X X_{i+1} + h_Y Y_{i+1}) ~
\label{trivial_rank2}
}
with even $N$ and nonzero $J_X, h_X,$ and $h_Y$.
The other is
\eq{
H &=& \sumi ( J_X (X_i X_{i+1} - Y_i Y_{i+1}) + \jz  + \hz ) ~, \nx
Q &=& \sumi (-1)^i Z_i ~,
\label{trivial_rank3}
}
with even $N$ and nonzero $J_X, J_Z,$ and $h_Z$.
We can easily confirm that these $(H,Q)$ are commutative.
In the following, we prove that there is no other pair of a non-integrable system and its short-support conserved quantity.

\subsection{Rank 1}
Short-support local conserved quantities of non-integrable systems with rank 1 Hamiltonian have been analyzed in Ref.~\cite{chiba2024proof}. In this previous study, it is proved that no short-support conserved quantity exists other than $H$ and $I$.

\subsection{Rank 2}
The standard form of the Hamiltonian in this case is 
\eq{
H = \sumi \loq{  \jx  +\jy \\ +\hx +\hy +\hz }
}
with $J_X,J_Y \neq 0$. 
We consider the case where the system is non-integrable ($(h_X,h_Y)\neq(0,0)$).

First, we explore $2$-support conserved quantities other than $H$.
By the same argument as Step 1 for $k \geq 3$, we obtain
\eq{
q_{(ZX)_i} = q_{(ZY)_i} = q_{(ZZ)_i} =  q_{(XZ)_i} = q_{(YZ)_i} &=& 0
} and \eq{
\frac{q_{(XX)_i}}{J_X} = \frac{q_{(YY)_{i+1}}}{J_Y} = \frac{q_{(XX)_{i+2}}}{J_X} &=& \cdots ~, \label{trivial_rank2_-1} 
\\
q_{(XY)_i} = -q_{(YX)_{i+1}} = q_{(XY)_{i+2}} &=& \cdots ~, \label{trivial_rank2_0}
}
for all $i$. 

Here, by adding a $2$-support conserved quantity $H$ with an appropriate factor if necessary, without loss of generality a $2$-support conserved quantity is assumed to satisfy
\eq{
q_{X_1 X_2} = - q_{X_2 X_3} ~.
}
Combining this with Eq.~\eqref{trivial_rank2_-1}, we have
\eq{
q_{(XX)_i} = (-1)^n q_{(XX)_{i+n}}
\label{trivial_rank2_assume}
}
for all $i$ and $n$.

Let us consider $(XZ)_i$ in $[Q,H]$.
This operator is generated by the following three commutator:
\eq{
\bmat{
 & X_i & Y_{i+1}  \\
 &     & X_{i+1}  \\ \hline
-& X_i & Z_{i+1}
}
\qquad
\bmat{
 & X_i & X_{i+1}  \\
 &     & Y_{i+1}  \\ \hline
+& X_i & Z_{i+1}
}
\qquad
\bmat{
 &     & Y_{i+1}  \\
 & X_i & X_{i+1}  \\ \hline
-& X_i & Z_{i+1}
}~,
}
which implies
\eq{
-h_X q_{(XY)_i} + h_Y q_{(XX)_i} - J_X q_{Y_{i+1}} = 0 ~.
\label{trivial_rank2_1}
}
Similarly, $(ZX)_{i+1}$ in $[Q,H]$ is generated by the following three commutators:
\eq{
\bmat{
 & Y_{i+1} & X_{i+2}  \\
 & X_{i+1} &  \\ \hline
-& Z_{i+1} & X_{i+2}
}
\qquad
\bmat{
 & X_{i+1} & X_{i+2}  \\
 & Y_{i+1} &  \\ \hline
+& Z_{i+1} & X_{i+2}
}
\qquad
\bmat{
 & Y_{i+1} &  \\
 & X_{i+1} & X_{i+2} \\ \hline
-& Z_{i+1} & X_{i+2}
}~, \nx
}
which implies
\eq{
-h_X q_{(YX)_{i+1}} + h_Y q_{(XX)_{i+1}} - J_X q_{Y_{i+1}} = 0 ~.
\label{trivial_rank2_2}
}
Substituting Eqs.~\eqref{trivial_rank2_0} and \eqref{trivial_rank2_assume} into Eq.~\eqref{trivial_rank2_2} and comparing with Eq.~\eqref{trivial_rank2_1}, we obtain
\eq{
\label{trivial_rank2_3}
-2h_X q_{(XY)_i} + 2h_Y q_{(XX)_i} &=& 0 ~, \\
\label{trivial_rank2_5}
q_{Y_{i+1}} &=& 0 ~.
}

To obtain further relations, let us consider $(YZ)_i$ and $(ZY)_{i+1}$ in $[Q,H]$.
$(YZ)_i$ is generated by
\eq{
\bmat{
 & Y_i & X_{i+1}  \\
 &     & Y_{i+1}  \\ \hline
+& Y_i & Z_{i+1}
}
\qquad
\bmat{
 & Y_i & Y_{i+1}  \\
 &     & X_{i+1}  \\ \hline
-& Y_i & Z_{i+1}
}
\qquad
\bmat{
 &     & X_{i+1}  \\
 & Y_i & Y_{i+1}  \\ \hline
+& Y_i & Z_{i+1}
}~,
}
which implies 
\eq{
\label{trivial_rank2_1.5}
h_Y q_{(YX)_i} - h_X q_{(YY)_i} + q_Y h_{X_{i+1}} &=& 0 ~.
}
$(ZY)_{i+1}$ is generated by
\eq{
\bmat{
 & X_{i+1} & Y_{i+2}  \\
 & Y_{i+1} &  \\ \hline
+& Z_{i+1} & Y_{i+2}
}
\qquad
\bmat{
 & Y_{i+1} & Y_{i+2}  \\
 & X_{i+1} &  \\ \hline
-& Z_{i+1} & Y_{i+2}
}
\qquad
\bmat{
 & X_{i+1} &  \\
 & Y_{i+1} & Y_{i+2} \\ \hline
+& Z_{i+1} & X_{i+2}
} ~, \nx
}
which implies
\eq{
\label{trivial_rank2_2.5}
h_Y q_{(XY)_{i+1}} - h_X q_{(YY)_{i+1}} + J_Y q_{X_{i+1}} &=& 0 ~.
}
Substituting Eqs.~\eqref{trivial_rank2_0} and \eqref{trivial_rank2_assume} into Eq.~\eqref{trivial_rank2_2.5} and comparing with Eq.~\eqref{trivial_rank2_1.5}, we obtain
\eq{
\label{trivial_rank2_4}
2h_Y q_{(YX)_i} - 2h_X q_{(YY)_i} &=& 0~, \\
\label{trivial_rank2_6}
q_{X_{i+1}} &=& 0 ~.
}

Combining Eqs.~\eqref{trivial_rank2_0}, \eqref{trivial_rank2_3}, and \eqref{trivial_rank2_4}, we obtain
\eq{
&& q_{(XX)_i} : q_{(XY)_i} : q_{(YX)_{i+1}} : q_{(YY)_{i+1}} \nx
&=& h_X^2 : h_X h_Y : (-h_X h_Y) : (-h_Y^2) ~.
\label{trivial_rank2_7}
}
Moreover, combining Eqs.~\eqref{trivial_rank2_assume} and \eqref{trivial_rank2_7}, we obtain that all the remaining coefficients of $2$-support operators are connected by linear relations as
\eq{
&& q_{(XX)_i} : q_{(XY)_i} : q_{(YX)_i} : q_{(YY)_i} \nx
&&:q_{(XX)_{i+n}} : q_{(XY)_{i+n}} q_{(YX)_{i+n}} : q_{(YY)_{i+n}} \nx
&=& h_X^2 : h_X h_Y : h_X h_Y : h_Y^2 \nx
&&:(-1)^n h_X^2 : (-1)^n h_X h_Y : (-1)^n h_X h_Y : (-1)^n h_Y^2 \nx
\label{trivial_rank2_8}
}
for any $i$ and $n$.
Therefore, $q_{(XX)_i}$ must be nonzero for a $2$-support conserved quantity independent of $H$.
Due to Eq.~\eqref{trivial_rank2_-1}, the existence of a 2-support conserved quantity independent of $H$ requires
\eq{
J_X : J_Y = h_X^2 : (-h_Y^2) ~,
}
which also indicates that all of $q_{(XX)_i}$, $q_{(XY)_i}$, $q_{(YX)_i}$, and $q_{(YY)_i}$ are nonzero.

Meanwhile,
\eq{
\bmat{
 & Y_i & X_{i+1} \\
 & Z_i &  \\ \hline
+& X_i & X_{i+1}
}
\qquad
\bmat{
 & X_i & Y_{i+1} \\
 &     & Z_{i+1} \\ \hline
+& X_i & X_{i+1}
}
}
yields
\eq{
h_Z q_{(YX)_i} + h_Z q_{(XY)_i} = 0 ~.
}
Since $q_{(XY)_i}(=q_{(YX)_i})$ is nonzero,
we find that
\eq{h_Z = 0}
should be satisfied if a $2$-support conserved quantity exists other than $H$.

Finally, we examine $q_{Z_i}$.
Due to Eqs.~\eqref{trivial_rank2_5} and \eqref{trivial_rank2_6},
the only commutator which generates $Y_i$ is
\eq{
\bmat{
&Z_i \\ &X_i \\ \hline +&Y_i
}~,
}
which implies 
\eq{
h_X q_{Z_i} &=& 0 ~.}
The only commutator which generates $X_i$ is
\eq{
\bmat{
&Z_i \\ &Y_i \\ \hline -&X_i
}~,
}
which implies
\eq{
-h_Y q_{Z_i} &=& 0 ~.
}
As declared, we treat non-integrable systems, implying $(h_X,h_Y) \neq (0,0)$, which suggests
\eq{
q_{Z_i} = 0 ~.
}
Combining this with Eqs.~\eqref{trivial_rank2_5} and \eqref{trivial_rank2_6}, we obtain that all the $1$-support operators of the $2$-support conserved quantity we discuss here have zero coefficients.

To sum up,
the only possible pair of rank 2 Hamiltonian $H$ and its 2-local conserved quantity $Q$ is
\eq{
H &=& \sumi (\jx + \jy + \hx + \hy) ~,\nx
Q &=& \sumi (-1)^i \loq{ h_X^2 X_i X_{i+1} + h_X h_Y X_i Y_{i+1} \\
+ h_X h_Y Y_i X_{i+1} + h_Y^2 Y_i Y_{i+1} }
}
with
\eq{
J_Y = - J_X \(\frac{h_Y}{h_X}\)^2 ~.
}
This pair is what we have exhibited in Eq.~\eqref{trivial_rank2}.

Finally, we demonstrate that no $1$-support conserved quantity exist.
First, the commutator
\eq{
\bmat{
&X_i \\ &Y_i & Y_{i+1} \\ \hline +&Z_i & Y_{i+1}
}
}
gives 
\eq{
q_{X_i} = 0~,
}
and the commutator
\eq{
\bmat{
&Y_i \\ &X_i & X_{i+1} \\ \hline -&Z_i & X_{i+1}
}
}
gives 
\eq{
q_{Y_i} = 0~.
}
In addition, commutators
\eq{
\bmat{ &Z_i \\ &X_i \\ \hline +&Y_i}
}
for $h_X\neq 0$ and 
\eq{
\bmat{ &Z_i \\ &Y_i \\ \hline -&X_i}
}
for $h_Y\neq 0$ give 
\eq{
q_{Z_i} = 0~,
}
which completes the proof for the absence of $1$-support conserved quantity.

\subsection{Rank 3}
The standard form of the Hamiltonian in this case is 
\eq{
H = \sumi \loq{  \jx  +\jy + \jz\\ +\hx +\hy +\hz }
}
with $J_X,J_Y,J_Z \neq 0$. 
We focus on the case where the system is non-integrable, which means that it is neither the XXX model, the XXZ model, nor the XYZ model.

First, we prove that rank 3 Hamiltonian cannot have $2$-support conserved quantities independent of $H$.

Since the only commutator generating $(XXZ)_i$ in $[Q,H]$ is
\eq{
\bmat{
& X_i & Y_{i+1} \\
&     & Z_{i+1} & Z_{i+2} \\ \hline
+&X_i & X_{i+1} & Z_{i+2}
}~,
}
we obtain
\eq{
J_Z q_{(XY)_i} = 0 ~.
}
Similarly, we find
\eq{
&& J_{(XY)_i} = J_{(YZ)_i} = J_{(ZX)_i} \nx
&=&J_{(YX)_i} = J_{(ZY)_i} = J_{(XZ)_i} = 0 ~.
}

Commutators generating $(XZY)_i$ are
\eq{
\bmat{
& X_i & X_{i+1} \\
&     & Y_{i+1} & Y_{i+2} \\ \hline
+&X_i & Z_{i+1} & Y_{i+2}
}\qquad
\bmat{
&     & Y_{i+1} & Y_{i+2}\\
& X_i & X_{i+1} &  \\ \hline
-&X_i & Z_{i+1} & Y_{i+2}
} ~,
}
which gives the following equation:
\eq{
J_Y q_{(XX)_i} - J_X q_{(YY)_{i+1}} = 0 ~.
\label{trivial_rank3_a}
}
Commutators generating $(ZXY)_i$ are
\eq{
\bmat{
& Z_i & Z_{i+1} \\
&     & Y_{i+1} & Y_{i+2} \\ \hline
-&Z_i & X_{i+1} & Y_{i+2}
}\qquad
\bmat{
&     & Y_{i+1} & Y_{i+2}\\
& Z_i & Z_{i+1} &  \\ \hline
+&Z_i & X_{i+1} & Y_{i+2}
}~,
}
which gives the following equation:
\eq{
-J_Y q_{(ZZ)_i} + J_Z q_{(YY)_{i+1}} = 0 ~.
\label{trivial_rank3_b}
}
Combining Eqs.~\eqref{trivial_rank3_a} and \eqref{trivial_rank3_b}, we find 
\eq{
q_{(XX)_i} : q_{(YY)_i} : q_{(ZZ)_i} = J_X : J_Y : J_Z ~,
}
which is the same as the Hamiltonian itself.
This means that, by adding $H$ to any 2-support conserved quantity $Q$ with an appropriate factor, all coefficients of 2-support operators of $Q$ can be eliminated.
In other words, there is no $2$-support conserved quantity independent of $H$.

Then we explore $1$-support conserved quantities.
Commutators generating $Z_i$ are expressed as
\eq{
\bmat{
&X_i \\
&Y_i \\ \hline
+&Z_i
}\qquad
\bmat{
&Y_i \\
&X_i \\ \hline
-&Z_i
},
}
which yields the equation $h_Y q_{X_i} - h_X q_{Y_i} = 0$.
Similarly, considering commutators generating $X_i$, we obtain $h_Z q_{Y_i} - h_Y q_{Z_i} = 0$.
Combining them, we find
\eq{
q_{X_i} : q_{Y_i} : q_{Z_i} = h_X : h_Y : h_Z ~.
\label{trivial_rank3_ratio}
}

In addition, commutators generating $(ZY)_i$ are
\eq{
\bmat{
& X_i \\
& Y_i & Y_{i+1} \\ \hline
+&Z_i & Y_{i+1}
}\qquad
\bmat{
& & X_{i+1} \\
& Z_i& Z_{i+1} \\ \hline
-&Z_i & Y_{i+1}
},
}
which gives the equation
\eq{
J_Y q_{X_i} - J_Z q_{X_{i+1}} = 0 ~,
\label{trivial_rank3_1}
}
and commutators generating $(YZ)_i$ are
\eq{
\bmat{
& X_i \\
& Z_i & Z_{i+1} \\ \hline
-&Y_i & Z_{i+1}
}\qquad
\bmat{
& & X_{i+1} \\
& Y_i& Y_{i+1} \\ \hline
+&Y_i & Z_{i+1}
},
}
which gives the equation
\eq{
-J_Z q_{X_i} + J_Y q_{X_{i+1}} = 0 ~.
\label{trivial_rank3_2}
}
Comparing Eqs.~\eqref{trivial_rank3_1} and \eqref{trivial_rank3_2}, we find that either $|J_Y| = |J_Z|$ or $q_{X_i} = 0$ holds.
Similar arguments with exchanging $X$, $Y$, and $Z$ suggest the following three conditions
\eq{
|J_Y| = |J_Z| \quad&\text{or}&\quad q_{X_i} = 0 ~,\nx
|J_Z| = |J_X| \quad&\text{or}&\quad q_{Y_i} = 0 ~,\nx
|J_X| = |J_Y| \quad&\text{or}&\quad q_{Z_i} = 0 ~,
\label{trivial_rank3_or}
}
all of which should be satisfied.

Now we shall divide the cases by whether $|J_X|$, $|J_Y|$, and $|J_Z|$ have equal values and analyze these cases one by one.
First, we consider the case of $|J_X| = |J_Y| = |J_Z|$.
If $J_X = J_Y = J_Z$ holds, the system is an integrable system, the XXX model, which is out of our scope.
We consider the case that two of signs of $J_X$, $J_Y$, and $J_Z$ are equal and the other is different.
Due to the symmetry of $X, Y,$ and $Z$, without loss of generality, we treat only the case of $J_X = J_Y = -J_Z$.

From Eq.~\eqref{trivial_rank3_1} and similar equations, we have
\eq{
q_{X_{i+1}} &=& - q_{X_i} \nx
q_{Y_{i+1}} &=& - q_{Y_i} \nx
q_{Z_{i+1}} &=& q_{Z_i} ~.
}
Combining this with Eq.~\eqref{trivial_rank3_ratio}, we obtain
\eq{
q_{X_{i+1}} : q_{Y_{i+1}} : q_{Z_{i+1}} = (-h_X) : (-h_Y) : h_Z
}
for all $i$.
By comparing with Eq.~\eqref{trivial_rank3_ratio},
we find that a $1$-support conserved quantity should satisfy either $h_X=h_Y=0$ or $h_Z = 0$.
In the former case, the system is an integrable system, the XXZ model, which is out of our scope.
In the latter case, the following non-integrable Hamiltonian and $1$-support quantity are indeed commutative:
\eq{
H &=& \sumi \loq{ J_X (X_i X_{i+1} + Y_i Y_{i+1} - Z_i Z_{i+1}) \\ + h_X X_i + h_Y Y_i} ~, \nx
Q &=& \sumi (-1)^i ( \hx + \hy ) ~.
}
Applying a global spin rotation 
\eq{
R=\begin{pmatrix}
\frac{h_Y}{\sqrt{h_X^2+h_Y^2}}& -\frac{h_X}{\sqrt{h_X^2+h_Y^2}} &0\\ 
0&0&-1 \\
\frac{h_X}{\sqrt{h_X^2+h_Y^2}}& \frac{h_Y}{\sqrt{h_X^2+h_Y^2}}&0
\end{pmatrix},
}
we obtain the special case of Eq.~\eqref{trivial_rank3} with $J_X=J_Z$.

Next, we consider the case where two of $|J_X|$, $|J_Y|$, and $|J_Z|$ are the same and the rest is different.
Due to the symmetry of $X, Y,$ and $Z$, without loss of generality, we treat only the case of $|J_X| = |J_Y| \neq |J_Z|$.
In this setup, Eq.~\eqref{trivial_rank3_or} implies $q_{X_i} = q_{Y_i} = 0$, and Eq.~\eqref{trivial_rank3_ratio} further leads to $h_X = h_Y = 0$.
If $J_X = J_Y$, the system is an integrable system, the XXZ model, which is out of our scope.
Now we clarify the condition $|J_X| = |J_Y|$.
If $J_X = - J_Y$, an equation similar to Eq.~\eqref{trivial_rank3_1} yields
\eq{
q_{Z_{i+1}} = - q_{Z_i} ~.
}
The following non-integrable Hamiltonian and $1$-support quantity are indeed commutative:
\eq{
H &=& \sumi (J_X (X_i X_{i+1} - Y_i Y_{i+1}) +\jz + \hz) ~, \nx
Q &=& \sumi (-1)^i Z_i ~.
}
This pair is what we have exhibited in Eq.~\eqref{trivial_rank3}.

Finally, if all of $|J_X|$, $|J_Y|$, and $|J_Z|$ differ, there are no $1$-support conserved quantities due to Eq.~\eqref{trivial_rank3_or}.
This completes the proof.

\section{Summarizing conclusion}\label{sec:Conclusion}
In this paper, we perform a comprehensive analysis of local conserved quantities for general spin-1/2 chains with shift-invariant and parity-symmetric nearest-neighbor interaction.
We prove the absence of nontrivial local conserved quantities for all models besides already known integrable systems. 
Our theorem brings a definitive end to the investigation of further integrable systems in this class of systems.
This also substantiates the expectation that the absence of local conserved quantities is a generic property in quantum many-body systems within the range of this class.
Our results establishes a sharp contrast between integrable systems, which has infinitely many local conserved quantities, and non-integrable systems, which has no nontrivial local conserved quantities. 
In addition, all short-support conserved quantities possessed by non-integrable systems are clarified.

As seen from its proof, our theorem of non-integrability is shown by proving the non-integrability of all possible models one by one.
To extend our result to more general setups, a unified and systematic proof technique of non-integrability is strongly desired.
Such a tool will bring a transparent view of integrability and non-integrability, which is left as a challenging future problem.

\bigskip

\subsection*{Acknowledgement}
The authors are grateful to \mbox{HaRu K. Park}, \mbox{Atsuo Kuniba}, \mbox{Chihiro Matsui}, \mbox{Kazuhiko Minami}, \mbox{Masaya Kunimi}, \mbox{Hal Tasaki}, and \mbox{Hosho Katsura} for fruitful discussion.
M.Y. is supported by WINGS-FMSP.
N.S. and Y.C. are supported by JST ERATO Grant No.~JPMJER2302, Japan.
Y.C. is also supported by Japan Society for the Promotion of Science KAKENHI Grant No.~JP21J14313 and the Special Postdoctoral Researchers Program at RIKEN.

\bibliography{main}

\providecommand{\noopsort}[1]{}\providecommand{\singleletter}[1]{#1}%
\begin{thebibliography}{50}%
\makeatletter
\providecommand \@ifxundefined [1]{%
 \@ifx{#1\undefined}
}%
\providecommand \@ifnum [1]{%
 \ifnum #1\expandafter \@firstoftwo
 \else \expandafter \@secondoftwo
 \fi
}%
\providecommand \@ifx [1]{%
 \ifx #1\expandafter \@firstoftwo
 \else \expandafter \@secondoftwo
 \fi
}%
\providecommand \natexlab [1]{#1}%
\providecommand \enquote  [1]{``#1''}%
\providecommand \bibnamefont  [1]{#1}%
\providecommand \bibfnamefont [1]{#1}%
\providecommand \citenamefont [1]{#1}%
\providecommand \href@noop [0]{\@secondoftwo}%
\providecommand \href [0]{\begingroup \@sanitize@url \@href}%
\providecommand \@href[1]{\@@startlink{#1}\@@href}%
\providecommand \@@href[1]{\endgroup#1\@@endlink}%
\providecommand \@sanitize@url [0]{\catcode `\\12\catcode `\$12\catcode `\&12\catcode `\#12\catcode `\^12\catcode `\_12\catcode `\%12\relax}%
\providecommand \@@startlink[1]{}%
\providecommand \@@endlink[0]{}%
\providecommand \url  [0]{\begingroup\@sanitize@url \@url }%
\providecommand \@url [1]{\endgroup\@href {#1}{\urlprefix }}%
\providecommand \urlprefix  [0]{URL }%
\providecommand \Eprint [0]{\href }%
\providecommand \doibase [0]{https://doi.org/}%
\providecommand \selectlanguage [0]{\@gobble}%
\providecommand \bibinfo  [0]{\@secondoftwo}%
\providecommand \bibfield  [0]{\@secondoftwo}%
\providecommand \translation [1]{[#1]}%
\providecommand \BibitemOpen [0]{}%
\providecommand \bibitemStop [0]{}%
\providecommand \bibitemNoStop [0]{.\EOS\space}%
\providecommand \EOS [0]{\spacefactor3000\relax}%
\providecommand \BibitemShut  [1]{\csname bibitem#1\endcsname}%
\let\auto@bib@innerbib\@empty
\bibitem [{\citenamefont {Baxter}(2016)}]{baxter2016exactly}%
  \BibitemOpen
  \bibfield  {author} {\bibinfo {author} {\bibfnamefont {R.~J.}\ \bibnamefont {Baxter}},\ }\bibfield  {title} {\bibinfo {title} {{Exactly solved models in statistical mechanics}},\ }\href@noop {} {\bibfield  {journal} {\bibinfo  {journal} {Elsevier}\ } (\bibinfo {year} {2016})}\BibitemShut {NoStop}%
\bibitem [{\citenamefont {Jimbo}\ and\ \citenamefont {Miwa}(1994)}]{jimbo1994algebraic}%
  \BibitemOpen
  \bibfield  {author} {\bibinfo {author} {\bibfnamefont {M.}~\bibnamefont {Jimbo}}\ and\ \bibinfo {author} {\bibfnamefont {T.}~\bibnamefont {Miwa}},\ }\bibfield  {title} {\bibinfo {title} {{Algebraic analysis of solvable lattice models}},\ }\href@noop {} {\bibfield  {journal} {\bibinfo  {journal} {American Mathematical Soc.}\ }\textbf {\bibinfo {volume} {85}} (\bibinfo {year} {1994})}\BibitemShut {NoStop}%
\bibitem [{\citenamefont {Takahashi}(1999)}]{takahashi1999thermodynamics}%
  \BibitemOpen
  \bibfield  {author} {\bibinfo {author} {\bibfnamefont {M.}~\bibnamefont {Takahashi}},\ }\bibfield  {title} {\bibinfo {title} {{Thermodynamics of one-dimensional solvable models}},\ }\href@noop {} {\bibfield  {journal} {\bibinfo  {journal} {Cambridge university press Cambridge}\ } (\bibinfo {year} {1999})}\BibitemShut {NoStop}%
\bibitem [{\citenamefont {Fagotti}\ and\ \citenamefont {Essler}(2013)}]{fagotti2013reduced}%
  \BibitemOpen
  \bibfield  {author} {\bibinfo {author} {\bibfnamefont {M.}~\bibnamefont {Fagotti}}\ and\ \bibinfo {author} {\bibfnamefont {F.~H.}\ \bibnamefont {Essler}},\ }\bibfield  {title} {\bibinfo {title} {{Reduced density matrix after a quantum quench}},\ }\href@noop {} {\bibfield  {journal} {\bibinfo  {journal} {Physical Review B—Condensed Matter and Materials Physics}\ }\textbf {\bibinfo {volume} {87}},\ \bibinfo {pages} {245107} (\bibinfo {year} {2013})}\BibitemShut {NoStop}%
\bibitem [{\citenamefont {Sklyanin}(1992)}]{sklyanin1992quantum}%
  \BibitemOpen
  \bibfield  {author} {\bibinfo {author} {\bibfnamefont {E.}~\bibnamefont {Sklyanin}},\ }\bibfield  {title} {\bibinfo {title} {{Quantum inverse scattering method. Selected topics}},\ }\href@noop {} {\bibfield  {journal} {\bibinfo  {journal} {arXiv preprint hep-th/9211111}\ } (\bibinfo {year} {1992})}\BibitemShut {NoStop}%
\bibitem [{\citenamefont {Faddeev}(1996)}]{faddeev1996algebraic}%
  \BibitemOpen
  \bibfield  {author} {\bibinfo {author} {\bibfnamefont {L.}~\bibnamefont {Faddeev}},\ }\bibfield  {title} {\bibinfo {title} {{How algebraic Bethe ansatz works for integrable model}},\ }\href@noop {} {\bibfield  {journal} {\bibinfo  {journal} {arXiv preprint hep-th/9605187}\ } (\bibinfo {year} {1996})}\BibitemShut {NoStop}%
\bibitem [{\citenamefont {Tetel'man}(1982)}]{tetel1982lorentz}%
  \BibitemOpen
  \bibfield  {author} {\bibinfo {author} {\bibfnamefont {M.}~\bibnamefont {Tetel'man}},\ }\bibfield  {title} {\bibinfo {title} {{Lorentz group for two-dimensional integrable lattice systems}},\ }\href@noop {} {\bibfield  {journal} {\bibinfo  {journal} {Soviet Journal of Experimental and Theoretical Physics}\ }\textbf {\bibinfo {volume} {55}},\ \bibinfo {pages} {306} (\bibinfo {year} {1982})}\BibitemShut {NoStop}%
\bibitem [{\citenamefont {Thacker}(1986)}]{thacker1986corner}%
  \BibitemOpen
  \bibfield  {author} {\bibinfo {author} {\bibfnamefont {H.}~\bibnamefont {Thacker}},\ }\bibfield  {title} {\bibinfo {title} {{Corner transfer matrices and Lorentz invariance on a lattice}},\ }\href@noop {} {\bibfield  {journal} {\bibinfo  {journal} {Physica D: Nonlinear Phenomena}\ }\textbf {\bibinfo {volume} {18}},\ \bibinfo {pages} {348} (\bibinfo {year} {1986})}\BibitemShut {NoStop}%
\bibitem [{\citenamefont {Bethe}(1931)}]{1931.Bethe.ZP.71}%
  \BibitemOpen
  \bibfield  {author} {\bibinfo {author} {\bibfnamefont {H.~A.}\ \bibnamefont {Bethe}},\ }\bibfield  {title} {\bibinfo {title} {{Zur {T}heorie der {M}etalle. I. {E}igenwerte und {E}igenfunktionen der linearen {A}tomkette}},\ }\href@noop {} {\bibfield  {journal} {\bibinfo  {journal} {Zeit. f\"ur Physik}\ }\textbf {\bibinfo {volume} {71}},\ \bibinfo {pages} {205} (\bibinfo {year} {1931})}\BibitemShut {NoStop}%
\bibitem [{\citenamefont {Lieb}\ \emph {et~al.}(1961)\citenamefont {Lieb}, \citenamefont {Schultz},\ and\ \citenamefont {Mattis}}]{lieb1961two}%
  \BibitemOpen
  \bibfield  {author} {\bibinfo {author} {\bibfnamefont {E.}~\bibnamefont {Lieb}}, \bibinfo {author} {\bibfnamefont {T.}~\bibnamefont {Schultz}},\ and\ \bibinfo {author} {\bibfnamefont {D.}~\bibnamefont {Mattis}},\ }\bibfield  {title} {\bibinfo {title} {{Two soluble models of an antiferromagnetic chain}},\ }\href@noop {} {\bibfield  {journal} {\bibinfo  {journal} {Annals of Physics}\ }\textbf {\bibinfo {volume} {16}},\ \bibinfo {pages} {407} (\bibinfo {year} {1961})}\BibitemShut {NoStop}%
\bibitem [{\citenamefont {Yang}\ and\ \citenamefont {Yang}(1966)}]{yang1966one}%
  \BibitemOpen
  \bibfield  {author} {\bibinfo {author} {\bibfnamefont {C.-N.}\ \bibnamefont {Yang}}\ and\ \bibinfo {author} {\bibfnamefont {C.-P.}\ \bibnamefont {Yang}},\ }\bibfield  {title} {\bibinfo {title} {{One-dimensional chain of anisotropic spin-spin interactions. I. Proof of Bethe's hypothesis for ground state in a finite system}},\ }\href@noop {} {\bibfield  {journal} {\bibinfo  {journal} {Physical Review}\ }\textbf {\bibinfo {volume} {150}},\ \bibinfo {pages} {321} (\bibinfo {year} {1966})}\BibitemShut {NoStop}%
\bibitem [{\citenamefont {Lieb}\ and\ \citenamefont {Wu}(1968)}]{lieb1968absence}%
  \BibitemOpen
  \bibfield  {author} {\bibinfo {author} {\bibfnamefont {E.~H.}\ \bibnamefont {Lieb}}\ and\ \bibinfo {author} {\bibfnamefont {F.-Y.}\ \bibnamefont {Wu}},\ }\bibfield  {title} {\bibinfo {title} {{Absence of Mott transition in an exact solution of the short-range, one-band model in one dimension}},\ }\href@noop {} {\bibfield  {journal} {\bibinfo  {journal} {Physical Review Letters}\ }\textbf {\bibinfo {volume} {20}},\ \bibinfo {pages} {1445} (\bibinfo {year} {1968})}\BibitemShut {NoStop}%
\bibitem [{\citenamefont {Baxter}(1971)}]{baxter1971one}%
  \BibitemOpen
  \bibfield  {author} {\bibinfo {author} {\bibfnamefont {R.}~\bibnamefont {Baxter}},\ }\bibfield  {title} {\bibinfo {title} {{One-dimensional anisotropic Heisenberg chain}},\ }\href@noop {} {\bibfield  {journal} {\bibinfo  {journal} {Physical Review Letters}\ }\textbf {\bibinfo {volume} {26}},\ \bibinfo {pages} {834} (\bibinfo {year} {1971})}\BibitemShut {NoStop}%
\bibitem [{\citenamefont {Maassarani}(1998)}]{maassarani1998xxc}%
  \BibitemOpen
  \bibfield  {author} {\bibinfo {author} {\bibfnamefont {Z.}~\bibnamefont {Maassarani}},\ }\bibfield  {title} {\bibinfo {title} {{The XXC models}},\ }\href@noop {} {\bibfield  {journal} {\bibinfo  {journal} {Physics Letters A}\ }\textbf {\bibinfo {volume} {244}},\ \bibinfo {pages} {160} (\bibinfo {year} {1998})}\BibitemShut {NoStop}%
\bibitem [{\citenamefont {Kitaev}(2001)}]{kitaev2001unpaired}%
  \BibitemOpen
  \bibfield  {author} {\bibinfo {author} {\bibfnamefont {A.~Y.}\ \bibnamefont {Kitaev}},\ }\bibfield  {title} {\bibinfo {title} {{Unpaired Majorana fermions in quantum wires}},\ }\href@noop {} {\bibfield  {journal} {\bibinfo  {journal} {Physics-uspekhi}\ }\textbf {\bibinfo {volume} {44}},\ \bibinfo {pages} {131} (\bibinfo {year} {2001})}\BibitemShut {NoStop}%
\bibitem [{\citenamefont {Yanagihara}\ and\ \citenamefont {Minami}(2020)}]{yanagihara2020exact}%
  \BibitemOpen
  \bibfield  {author} {\bibinfo {author} {\bibfnamefont {Y.}~\bibnamefont {Yanagihara}}\ and\ \bibinfo {author} {\bibfnamefont {K.}~\bibnamefont {Minami}},\ }\bibfield  {title} {\bibinfo {title} {{Exact solution of a cluster model with next-nearest-neighbor interaction}},\ }\href@noop {} {\bibfield  {journal} {\bibinfo  {journal} {Progress of Theoretical and Experimental Physics}\ }\textbf {\bibinfo {volume} {2020}},\ \bibinfo {pages} {113A01} (\bibinfo {year} {2020})}\BibitemShut {NoStop}%
\bibitem [{\citenamefont {Sutherland}(1975)}]{sutherland1975model}%
  \BibitemOpen
  \bibfield  {author} {\bibinfo {author} {\bibfnamefont {B.}~\bibnamefont {Sutherland}},\ }\bibfield  {title} {\bibinfo {title} {{Model for a multicomponent quantum system}},\ }\href@noop {} {\bibfield  {journal} {\bibinfo  {journal} {Physical Review B}\ }\textbf {\bibinfo {volume} {12}},\ \bibinfo {pages} {3795} (\bibinfo {year} {1975})}\BibitemShut {NoStop}%
\bibitem [{\citenamefont {Takhtajan}(1982)}]{Takhtajan:1982jeo}%
  \BibitemOpen
  \bibfield  {author} {\bibinfo {author} {\bibfnamefont {L.~A.}\ \bibnamefont {Takhtajan}},\ }\bibfield  {title} {\bibinfo {title} {{The picture of low-lying excitations in the isotropic Heisenberg chain of arbitrary spins}},\ }\href@noop {} {\bibfield  {journal} {\bibinfo  {journal} {Phys. Lett. A}\ }\textbf {\bibinfo {volume} {87}},\ \bibinfo {pages} {479} (\bibinfo {year} {1982})}\BibitemShut {NoStop}%
\bibitem [{\citenamefont {Babujian}(1982)}]{babujian1982exact}%
  \BibitemOpen
  \bibfield  {author} {\bibinfo {author} {\bibfnamefont {H.}~\bibnamefont {Babujian}},\ }\bibfield  {title} {\bibinfo {title} {{Exact solution of the one-dimensional isotropic Heisenberg chain with arbitrary spins S}},\ }\href@noop {} {\bibfield  {journal} {\bibinfo  {journal} {Physics Letters A}\ }\textbf {\bibinfo {volume} {90}},\ \bibinfo {pages} {479} (\bibinfo {year} {1982})}\BibitemShut {NoStop}%
\bibitem [{\citenamefont {Barber}\ and\ \citenamefont {Batchelor}(1989)}]{barber1989spectrum}%
  \BibitemOpen
  \bibfield  {author} {\bibinfo {author} {\bibfnamefont {M.~N.}\ \bibnamefont {Barber}}\ and\ \bibinfo {author} {\bibfnamefont {M.~T.}\ \bibnamefont {Batchelor}},\ }\bibfield  {title} {\bibinfo {title} {{Spectrum of the biquadratic spin-1 antiferromagnetic chain}},\ }\href@noop {} {\bibfield  {journal} {\bibinfo  {journal} {Physical Review B}\ }\textbf {\bibinfo {volume} {40}},\ \bibinfo {pages} {4621} (\bibinfo {year} {1989})}\BibitemShut {NoStop}%
\bibitem [{\citenamefont {Mori}\ \emph {et~al.}(2018)\citenamefont {Mori}, \citenamefont {Ikeda}, \citenamefont {Kaminishi},\ and\ \citenamefont {Ueda}}]{mori2018thermalization}%
  \BibitemOpen
  \bibfield  {author} {\bibinfo {author} {\bibfnamefont {T.}~\bibnamefont {Mori}}, \bibinfo {author} {\bibfnamefont {T.~N.}\ \bibnamefont {Ikeda}}, \bibinfo {author} {\bibfnamefont {E.}~\bibnamefont {Kaminishi}},\ and\ \bibinfo {author} {\bibfnamefont {M.}~\bibnamefont {Ueda}},\ }\bibfield  {title} {\bibinfo {title} {{Thermalization and prethermalization in isolated quantum systems: a theoretical overview}},\ }\href@noop {} {\bibfield  {journal} {\bibinfo  {journal} {Journal of Physics B: Atomic, Molecular and Optical Physics}\ }\textbf {\bibinfo {volume} {51}},\ \bibinfo {pages} {112001} (\bibinfo {year} {2018})}\BibitemShut {NoStop}%
\bibitem [{\citenamefont {Shiraishi}(2019)}]{shiraishi2019proof}%
  \BibitemOpen
  \bibfield  {author} {\bibinfo {author} {\bibfnamefont {N.}~\bibnamefont {Shiraishi}},\ }\bibfield  {title} {\bibinfo {title} {{Proof of the absence of local conserved quantities in the XYZ chain with a magnetic field}},\ }\href@noop {} {\bibfield  {journal} {\bibinfo  {journal} {Europhysics Letters}\ }\textbf {\bibinfo {volume} {128}},\ \bibinfo {pages} {17002} (\bibinfo {year} {2019})}\BibitemShut {NoStop}%
\bibitem [{\citenamefont {Dyson}(1962)}]{dyson1962statistical}%
  \BibitemOpen
  \bibfield  {author} {\bibinfo {author} {\bibfnamefont {F.~J.}\ \bibnamefont {Dyson}},\ }\bibfield  {title} {\bibinfo {title} {{Statistical theory of the energy levels of complex systems. I}},\ }\href@noop {} {\bibfield  {journal} {\bibinfo  {journal} {Journal of Mathematical Physics}\ }\textbf {\bibinfo {volume} {3}},\ \bibinfo {pages} {140} (\bibinfo {year} {1962})}\BibitemShut {NoStop}%
\bibitem [{\citenamefont {Wigner}(1993)}]{wigner1993characteristic}%
  \BibitemOpen
  \bibfield  {author} {\bibinfo {author} {\bibfnamefont {E.~P.}\ \bibnamefont {Wigner}},\ }\bibfield  {title} {\bibinfo {title} {{Characteristic vectors of bordered matrices with infinite dimensions I}},\ }\href@noop {} {\bibfield  {journal} {\bibinfo  {journal} {The Collected Works of Eugene Paul Wigner: Part A: The Scientific Papers}\ ,\ \bibinfo {pages} {524}} (\bibinfo {year} {1993})}\BibitemShut {NoStop}%
\bibitem [{\citenamefont {Rigol}(2009)}]{rigol2009breakdown}%
  \BibitemOpen
  \bibfield  {author} {\bibinfo {author} {\bibfnamefont {M.}~\bibnamefont {Rigol}},\ }\bibfield  {title} {\bibinfo {title} {{Breakdown of thermalization in finite one-dimensional systems}},\ }\href@noop {} {\bibfield  {journal} {\bibinfo  {journal} {Physical review letters}\ }\textbf {\bibinfo {volume} {103}},\ \bibinfo {pages} {100403} (\bibinfo {year} {2009})}\BibitemShut {NoStop}%
\bibitem [{\citenamefont {Santos}\ and\ \citenamefont {Rigol}(2010)}]{santos2010onset}%
  \BibitemOpen
  \bibfield  {author} {\bibinfo {author} {\bibfnamefont {L.~F.}\ \bibnamefont {Santos}}\ and\ \bibinfo {author} {\bibfnamefont {M.}~\bibnamefont {Rigol}},\ }\bibfield  {title} {\bibinfo {title} {{Onset of quantum chaos in one-dimensional bosonic and fermionic systems and its relation to thermalization}},\ }\href@noop {} {\bibfield  {journal} {\bibinfo  {journal} {Physical Review E—Statistical, Nonlinear, and Soft Matter Physics}\ }\textbf {\bibinfo {volume} {81}},\ \bibinfo {pages} {036206} (\bibinfo {year} {2010})}\BibitemShut {NoStop}%
\bibitem [{\citenamefont {Atas}\ \emph {et~al.}(2013)\citenamefont {Atas}, \citenamefont {Bogomolny}, \citenamefont {Giraud},\ and\ \citenamefont {Roux}}]{atas2013distribution}%
  \BibitemOpen
  \bibfield  {author} {\bibinfo {author} {\bibfnamefont {Y.~Y.}\ \bibnamefont {Atas}}, \bibinfo {author} {\bibfnamefont {E.}~\bibnamefont {Bogomolny}}, \bibinfo {author} {\bibfnamefont {O.}~\bibnamefont {Giraud}},\ and\ \bibinfo {author} {\bibfnamefont {G.}~\bibnamefont {Roux}},\ }\bibfield  {title} {\bibinfo {title} {{Distribution of the ratio of consecutive level spacings in random matrix ensembles}},\ }\href@noop {} {\bibfield  {journal} {\bibinfo  {journal} {Physical review letters}\ }\textbf {\bibinfo {volume} {110}},\ \bibinfo {pages} {084101} (\bibinfo {year} {2013})}\BibitemShut {NoStop}%
\bibitem [{\citenamefont {Chiba}(2024)}]{chiba2024proof}%
  \BibitemOpen
  \bibfield  {author} {\bibinfo {author} {\bibfnamefont {Y.}~\bibnamefont {Chiba}},\ }\bibfield  {title} {\bibinfo {title} {{Proof of absence of local conserved quantities in the mixed-field Ising chain}},\ }\href@noop {} {\bibfield  {journal} {\bibinfo  {journal} {Physical Review B}\ }\textbf {\bibinfo {volume} {109}},\ \bibinfo {pages} {035123} (\bibinfo {year} {2024})}\BibitemShut {NoStop}%
\bibitem [{\citenamefont {Park}\ and\ \citenamefont {Lee}(2024)}]{park2024proof}%
  \BibitemOpen
  \bibfield  {author} {\bibinfo {author} {\bibfnamefont {H.~K.}\ \bibnamefont {Park}}\ and\ \bibinfo {author} {\bibfnamefont {S.}~\bibnamefont {Lee}},\ }\bibfield  {title} {\bibinfo {title} {{Proof of the nonintegrability of PXP model and general spin-$1/2$ systems}},\ }\href@noop {} {\bibfield  {journal} {\bibinfo  {journal} {arXiv preprint arXiv:2403.02335}\ } (\bibinfo {year} {2024})}\BibitemShut {NoStop}%
\bibitem [{\citenamefont {Shiraishi}(2024)}]{shiraishi2024absence}%
  \BibitemOpen
  \bibfield  {author} {\bibinfo {author} {\bibfnamefont {N.}~\bibnamefont {Shiraishi}},\ }\bibfield  {title} {\bibinfo {title} {{Absence of local conserved quantity in the Heisenberg model with next-nearest-neighbor interaction}},\ }\href@noop {} {\bibfield  {journal} {\bibinfo  {journal} {Journal of Statistical Physics}\ }\textbf {\bibinfo {volume} {191}},\ \bibinfo {pages} {114} (\bibinfo {year} {2024})}\BibitemShut {NoStop}%
\bibitem [{\citenamefont {Neumann}(1929)}]{neumann1929beweis}%
  \BibitemOpen
  \bibfield  {author} {\bibinfo {author} {\bibfnamefont {J.}~\bibnamefont {Neumann}},\ }\bibfield  {title} {\bibinfo {title} {{Beweis des Ergodensatzes und des H-Theorems in der neuen Mechanik}},\ }\href@noop {} {\bibfield  {journal} {\bibinfo  {journal} {Zeit. f\"ur Physik}\ }\textbf {\bibinfo {volume} {58}},\ \bibinfo {pages} {30} (\bibinfo {year} {1929})}\BibitemShut {NoStop}%
\bibitem [{\citenamefont {Goldstein}\ \emph {et~al.}(2010)\citenamefont {Goldstein}, \citenamefont {Lebowitz}, \citenamefont {Tumulka},\ and\ \citenamefont {Zangh{\`\i}}}]{goldstein2010long}%
  \BibitemOpen
  \bibfield  {author} {\bibinfo {author} {\bibfnamefont {S.}~\bibnamefont {Goldstein}}, \bibinfo {author} {\bibfnamefont {J.~L.}\ \bibnamefont {Lebowitz}}, \bibinfo {author} {\bibfnamefont {R.}~\bibnamefont {Tumulka}},\ and\ \bibinfo {author} {\bibfnamefont {N.}~\bibnamefont {Zangh{\`\i}}},\ }\bibfield  {title} {\bibinfo {title} {{Long-time behavior of macroscopic quantum systems: Commentary accompanying the English translation of John von Neumann’s 1929 article on the quantum ergodic theorem}},\ }\href@noop {} {\bibfield  {journal} {\bibinfo  {journal} {The European Physical Journal H}\ }\textbf {\bibinfo {volume} {35}},\ \bibinfo {pages} {173} (\bibinfo {year} {2010})}\BibitemShut {NoStop}%
\bibitem [{\citenamefont {Kubo}(1957)}]{kubo1957statistical}%
  \BibitemOpen
  \bibfield  {author} {\bibinfo {author} {\bibfnamefont {R.}~\bibnamefont {Kubo}},\ }\bibfield  {title} {\bibinfo {title} {{Statistical-mechanical theory of irreversible processes. I. General theory and simple applications to magnetic and conduction problems}},\ }\href@noop {} {\bibfield  {journal} {\bibinfo  {journal} {Journal of the physical society of Japan}\ }\textbf {\bibinfo {volume} {12}},\ \bibinfo {pages} {570} (\bibinfo {year} {1957})}\BibitemShut {NoStop}%
\bibitem [{\citenamefont {Fourier}(1822)}]{fourier1822}%
  \BibitemOpen
  \bibfield  {author} {\bibinfo {author} {\bibfnamefont {J.}~\bibnamefont {Fourier}},\ }\bibfield  {title} {\bibinfo {title} {{Th\'eorie Analytique De La Chaleur}},\ }\href@noop {} {\bibfield  {journal} {\bibinfo  {journal} {Firmin Didot P\`ere et Fils.}\ } (\bibinfo {year} {1822})}\BibitemShut {NoStop}%
\bibitem [{\citenamefont {Mazur}(1969)}]{mazur1969non}%
  \BibitemOpen
  \bibfield  {author} {\bibinfo {author} {\bibfnamefont {P.}~\bibnamefont {Mazur}},\ }\bibfield  {title} {\bibinfo {title} {{Non-ergodicity of phase functions in certain systems}},\ }\href@noop {} {\bibfield  {journal} {\bibinfo  {journal} {Physica}\ }\textbf {\bibinfo {volume} {43}},\ \bibinfo {pages} {533} (\bibinfo {year} {1969})}\BibitemShut {NoStop}%
\bibitem [{\citenamefont {Suzuki}(1971)}]{suzuki1971ergodicity}%
  \BibitemOpen
  \bibfield  {author} {\bibinfo {author} {\bibfnamefont {M.}~\bibnamefont {Suzuki}},\ }\bibfield  {title} {\bibinfo {title} {{Ergodicity, constants of motion, and bounds for susceptibilities}},\ }\href@noop {} {\bibfield  {journal} {\bibinfo  {journal} {Physica}\ }\textbf {\bibinfo {volume} {51}},\ \bibinfo {pages} {277} (\bibinfo {year} {1971})}\BibitemShut {NoStop}%
\bibitem [{\citenamefont {Zotos}\ \emph {et~al.}(1997)\citenamefont {Zotos}, \citenamefont {Naef},\ and\ \citenamefont {Prelovsek}}]{zotos1997transport}%
  \BibitemOpen
  \bibfield  {author} {\bibinfo {author} {\bibfnamefont {X.}~\bibnamefont {Zotos}}, \bibinfo {author} {\bibfnamefont {F.}~\bibnamefont {Naef}},\ and\ \bibinfo {author} {\bibfnamefont {P.}~\bibnamefont {Prelovsek}},\ }\bibfield  {title} {\bibinfo {title} {{Transport and conservation laws}},\ }\href@noop {} {\bibfield  {journal} {\bibinfo  {journal} {Physical Review B}\ }\textbf {\bibinfo {volume} {55}},\ \bibinfo {pages} {11029} (\bibinfo {year} {1997})}\BibitemShut {NoStop}%
\bibitem [{\citenamefont {Saito}(2003)}]{saito2003strong}%
  \BibitemOpen
  \bibfield  {author} {\bibinfo {author} {\bibfnamefont {K.}~\bibnamefont {Saito}},\ }\bibfield  {title} {\bibinfo {title} {{Strong evidence of normal heat conduction in a one-dimensional quantum system}},\ }\href@noop {} {\bibfield  {journal} {\bibinfo  {journal} {Europhysics Letters}\ }\textbf {\bibinfo {volume} {61}},\ \bibinfo {pages} {34} (\bibinfo {year} {2003})}\BibitemShut {NoStop}%
\bibitem [{\citenamefont {Rigol}\ \emph {et~al.}(2008)\citenamefont {Rigol}, \citenamefont {Dunjko},\ and\ \citenamefont {Olshanii}}]{rigol2008thermalization}%
  \BibitemOpen
  \bibfield  {author} {\bibinfo {author} {\bibfnamefont {M.}~\bibnamefont {Rigol}}, \bibinfo {author} {\bibfnamefont {V.}~\bibnamefont {Dunjko}},\ and\ \bibinfo {author} {\bibfnamefont {M.}~\bibnamefont {Olshanii}},\ }\bibfield  {title} {\bibinfo {title} {{Thermalization and its mechanism for generic isolated quantum systems}},\ }\href@noop {} {\bibfield  {journal} {\bibinfo  {journal} {Nature}\ }\textbf {\bibinfo {volume} {452}},\ \bibinfo {pages} {854} (\bibinfo {year} {2008})}\BibitemShut {NoStop}%
\bibitem [{\citenamefont {Jordan}\ and\ \citenamefont {Wigner}(1928)}]{Jordan1928berDP}%
  \BibitemOpen
  \bibfield  {author} {\bibinfo {author} {\bibfnamefont {P.}~\bibnamefont {Jordan}}\ and\ \bibinfo {author} {\bibfnamefont {E.}~\bibnamefont {Wigner}},\ }\bibfield  {title} {\bibinfo {title} {{{\"U}ber das Paulische {\"A}quivalenzverbot}},\ }\href@noop {} {\bibfield  {journal} {\bibinfo  {journal} {Zeitschrift f{\"u}r Physik}\ }\textbf {\bibinfo {volume} {47}},\ \bibinfo {pages} {631} (\bibinfo {year} {1928})}\BibitemShut {NoStop}%
\bibitem [{\citenamefont {Grabowski}\ and\ \citenamefont {Mathieu}(1994)}]{grabowski1994quantum}%
  \BibitemOpen
  \bibfield  {author} {\bibinfo {author} {\bibfnamefont {M.~P.}\ \bibnamefont {Grabowski}}\ and\ \bibinfo {author} {\bibfnamefont {P.}~\bibnamefont {Mathieu}},\ }\bibfield  {title} {\bibinfo {title} {{Quantum integrals of motion for the Heisenberg spin chain}},\ }\href@noop {} {\bibfield  {journal} {\bibinfo  {journal} {Modern Physics Letters A}\ }\textbf {\bibinfo {volume} {9}},\ \bibinfo {pages} {2197} (\bibinfo {year} {1994})}\BibitemShut {NoStop}%
\bibitem [{\citenamefont {Nozawa}\ and\ \citenamefont {Fukai}(2020)}]{nozawa2020explicit}%
  \BibitemOpen
  \bibfield  {author} {\bibinfo {author} {\bibfnamefont {Y.}~\bibnamefont {Nozawa}}\ and\ \bibinfo {author} {\bibfnamefont {K.}~\bibnamefont {Fukai}},\ }\bibfield  {title} {\bibinfo {title} {{Explicit construction of local conserved quantities in the XYZ spin-1/2 chain}},\ }\href@noop {} {\bibfield  {journal} {\bibinfo  {journal} {Physical Review Letters}\ }\textbf {\bibinfo {volume} {125}},\ \bibinfo {pages} {090602} (\bibinfo {year} {2020})}\BibitemShut {NoStop}%
\bibitem [{\citenamefont {Grabowski}\ and\ \citenamefont {Mathieu}(1995)}]{grabowski1995structure}%
  \BibitemOpen
  \bibfield  {author} {\bibinfo {author} {\bibfnamefont {M.}~\bibnamefont {Grabowski}}\ and\ \bibinfo {author} {\bibfnamefont {P.}~\bibnamefont {Mathieu}},\ }\bibfield  {title} {\bibinfo {title} {{Structure of the conservation laws in quantum integrable spin chains with short range interactions}},\ }\href@noop {} {\bibfield  {journal} {\bibinfo  {journal} {Annals of Physics}\ }\textbf {\bibinfo {volume} {243}},\ \bibinfo {pages} {299} (\bibinfo {year} {1995})}\BibitemShut {NoStop}%
\bibitem [{\citenamefont {Fukai}(2023)}]{fukai2023all}%
  \BibitemOpen
  \bibfield  {author} {\bibinfo {author} {\bibfnamefont {K.}~\bibnamefont {Fukai}},\ }\bibfield  {title} {\bibinfo {title} {{All local conserved quantities of the one-dimensional Hubbard model}},\ }\href@noop {} {\bibfield  {journal} {\bibinfo  {journal} {Physical Review Letters}\ }\textbf {\bibinfo {volume} {131}},\ \bibinfo {pages} {256704} (\bibinfo {year} {2023})}\BibitemShut {NoStop}%
\bibitem [{\citenamefont {Caux}\ and\ \citenamefont {Mossel}(2011)}]{caux2011remarks}%
  \BibitemOpen
  \bibfield  {author} {\bibinfo {author} {\bibfnamefont {J.-S.}\ \bibnamefont {Caux}}\ and\ \bibinfo {author} {\bibfnamefont {J.}~\bibnamefont {Mossel}},\ }\bibfield  {title} {\bibinfo {title} {{Remarks on the notion of quantum integrability}},\ }\href@noop {} {\bibfield  {journal} {\bibinfo  {journal} {Journal of Statistical Mechanics: Theory and Experiment}\ }\textbf {\bibinfo {volume} {2011}},\ \bibinfo {pages} {P02023} (\bibinfo {year} {2011})}\BibitemShut {NoStop}%
\bibitem [{\citenamefont {Gogolin}\ and\ \citenamefont {Eisert}(2016)}]{gogolin2016equilibration}%
  \BibitemOpen
  \bibfield  {author} {\bibinfo {author} {\bibfnamefont {C.}~\bibnamefont {Gogolin}}\ and\ \bibinfo {author} {\bibfnamefont {J.}~\bibnamefont {Eisert}},\ }\bibfield  {title} {\bibinfo {title} {{Equilibration, thermalisation, and the emergence of statistical mechanics in closed quantum systems}},\ }\href@noop {} {\bibfield  {journal} {\bibinfo  {journal} {Reports on Progress in Physics}\ }\textbf {\bibinfo {volume} {79}},\ \bibinfo {pages} {056001} (\bibinfo {year} {2016})}\BibitemShut {NoStop}%
\bibitem [{\citenamefont {Liouville}(1855)}]{Liouville1855}%
  \BibitemOpen
  \bibfield  {author} {\bibinfo {author} {\bibfnamefont {J.}~\bibnamefont {Liouville}},\ }\bibfield  {title} {\bibinfo {title} {{Note sur l'int\'egration des \'equations diff\'erentielles de la Dynamique}},\ }\href@noop {} {\bibfield  {journal} {\bibinfo  {journal} {Journal de Math\'ematiques Pures et Appliqu\'ees}\ ,\ \bibinfo {pages} {137}} (\bibinfo {year} {1855})}\BibitemShut {NoStop}%
\bibitem [{\citenamefont {Arnol'd}(1963)}]{arnol1963small}%
  \BibitemOpen
  \bibfield  {author} {\bibinfo {author} {\bibfnamefont {V.~I.}\ \bibnamefont {Arnol'd}},\ }\bibfield  {title} {\bibinfo {title} {{Small denominators and problems of stability of motion in classical and celestial mechanics}},\ }\href@noop {} {\bibfield  {journal} {\bibinfo  {journal} {Russian Mathematical Surveys}\ }\textbf {\bibinfo {volume} {18}},\ \bibinfo {pages} {85} (\bibinfo {year} {1963})}\BibitemShut {NoStop}%
\bibitem [{\citenamefont {Hamazaki}\ \emph {et~al.}(2016)\citenamefont {Hamazaki}, \citenamefont {Ikeda},\ and\ \citenamefont {Ueda}}]{hamazaki2016generalized}%
  \BibitemOpen
  \bibfield  {author} {\bibinfo {author} {\bibfnamefont {R.}~\bibnamefont {Hamazaki}}, \bibinfo {author} {\bibfnamefont {T.~N.}\ \bibnamefont {Ikeda}},\ and\ \bibinfo {author} {\bibfnamefont {M.}~\bibnamefont {Ueda}},\ }\bibfield  {title} {\bibinfo {title} {{Generalized Gibbs ensemble in a nonintegrable system with an extensive number of local symmetries}},\ }\href@noop {} {\bibfield  {journal} {\bibinfo  {journal} {Physical Review E}\ }\textbf {\bibinfo {volume} {93}},\ \bibinfo {pages} {032116} (\bibinfo {year} {2016})}\BibitemShut {NoStop}%
\bibitem [{Note1()}]{Note1}%
  \BibitemOpen
  \bibinfo {note} {Note that the specification of support is unique for $k(\leq N/2)$-support Pauli strings, under the rule that support is chosen as the shortest one.}\BibitemShut {Stop}%
\end{thebibliography}%
\end{document}